\documentclass[conference]{IEEEtran}
\IEEEoverridecommandlockouts
\usepackage{bm}
\usepackage{cite}
\usepackage{amsmath,amssymb,amsfonts}
\usepackage{algorithmic}
\usepackage{graphicx}
\usepackage{textcomp}
\usepackage{xcolor}
\usepackage[colorinlistoftodos]{todonotes}
\usepackage{url}

\newcommand*\mcchain{\mathrel{-\mkern-3mu{\circ}\mkern-3mu-}}

\usepackage{titlesec}

\titleformat{\paragraph}
  {\normalfont\normalsize\bfseries}
  {\theparagraph}
  {1em}
  {}

\titlespacing*{\paragraph}
  {0pt}
  {3.25ex plus 1ex minus .2ex}
  {1.5ex plus .2ex}

\allowdisplaybreaks

\def\BibTeX{{\rm B\kern-.05em{\sc i\kern-.025em b}\kern-.08em
    T\kern-.1667em\lower.7ex\hbox{E}\kern-.125emX}}
\usepackage{amsmath,amssymb,amsthm}
\usepackage{tikz,enumerate}
\usepackage{cite,mathtools}
\theoremstyle{plain}

\newtheorem{theorem}{Theorem} 
 
\newtheorem{corollary}{Corollary}

\newtheorem{lemma}{Lemma}

\theoremstyle{remark}

\newtheorem{remark}{Remark}

\theoremstyle{definition}
\newtheorem{definition}{Definition}

\begin{document}
\title{The Capacity Region for Classes of Sum-Broadcast Channels}
\author{%
 \IEEEauthorblockN{Amin Gohari, Yi Liu and Chandra Nair}
\IEEEauthorblockA{Department of Information Engineering\\
The Chinese University of Hong Kong\\
Sha Tin, N.T., Hong Kong\\
Email: \{agohari, ly023, chandra\}@ie.cuhk.edu.hk}}
\onecolumn

\pagestyle{plain}      
\pagenumbering{arabic} 

\maketitle

\begin{abstract}
We compute the capacity region of a sum of broadcast channels whose components are degraded, less-noisy, more-capable, deterministic, or semi-deterministic. 
We achieve this by showing that an auxiliary-receiver outer bound, previously introduced by some of the authors, matches Marton's inner bound. This result generalizes a previously known result for the sum of two reversely degraded broadcast channels due to El Gamal (1980). Moreover, we define a class of \textit{primary} broadcast channels and show an analogous result for the sum of primary broadcast channels.  
\end{abstract}

\section{Introduction}
As depicted in Figure~\ref{fig:1}, a two-receiver discrete-memoryless broadcast channel $(T_{YZ|X},\mathcal{X},\mathcal{Y},\mathcal{Z})$ is a mathematical model in which the sender aims to communicate a private message $M_1$ to receiver $Y$, a private message $M_2$ to receiver $Z$, and a common message $M_0$ to both receivers (see \cite[Chapters 5, 8]{gk12} for detailed definitions). This model was first proposed by Cover \cite{cov72}. However, a computable characterization of the capacity region remains a fundamental open problem in the field of network information theory. Let $\mathcal{C}(T)$ denote the capacity region for the broadcast channel $T_{YZ|X}$.

\tikzstyle{box}=[rectangle, draw, text centered]
\tikzstyle{line} = [draw, -latex']

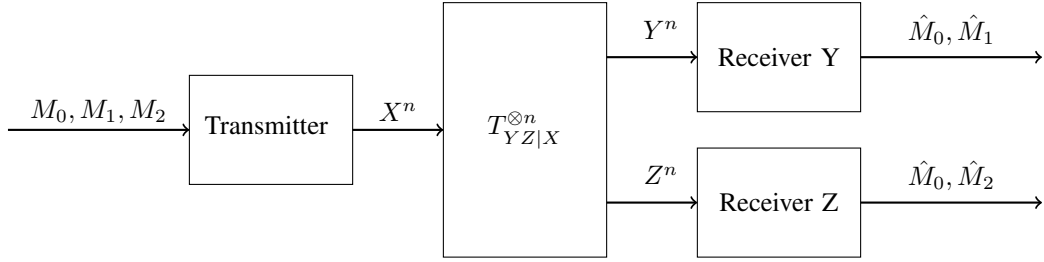
\begin{figure}[htb]
    \centering
    \begin{tikzpicture}[scale=1.2, every node/.style={scale=1}]
\node at (-0.2,-0.6) {$M_0, M_1,M_2$}; 
\draw [->,thick] (-1.2,-0.8) -- (0.8,-0.8);
\draw (0.8,-0.2) rectangle +(1.8,-1.2); \node at (1.65,-0.75) {Transmitter};
\draw [->,thick] (2.6,-0.8)-- (3.6,-0.8); \node at (3.1, -0.6) {$ X^n$};
\draw (3.6, -2.2) rectangle +(1.8,2.8); \node at (4.5,-0.8) {$T_{YZ|X}^{\otimes n}$};
\draw [->,thick] (5.4, 0) --(6.4, 0); \node at (6.0,0.3) {$Y^n$};
\draw [->,thick] (5.4, -1.6) --(6.4, -1.6); \node at (6.0,-1.3) {$Z^n$};
\draw (6.4,-0.6) rectangle +(1.8,1.2); \node at (7.3,0) {Receiver Y};
\draw (6.4,-2.2) rectangle +(1.8,1.2); \node at (7.3,-1.6) {Receiver Z};
\draw [->,thick] (8.2, 0) --(10.2,0); \node at (9.2, 0.3) {$\hat{ M}_0,\hat{ M}_1$};
\draw [->,thick] (8.2, -1.6) --(10.2,-1.6); \node at (9.2, -1.3) {$\hat{ M}_0,\hat{ M}_2$};
\end{tikzpicture}
    \caption{Two-receiver broadcast communication system.}
    \label{fig:1}
\end{figure}

\begin{theorem}
[Marton's inner bound, \cite{mar79}]
\label{thm:martonib}
The union of non-negative rate triples $(R_0, R_1, R_2)$ satisfying the constraints
\begin{subequations}
\begin{align}
    R_0 & \leq \min \{I(W;Y),I(W;Z) \},\label{eqnM1}\\ 
    R_0+R_1 & \leq I(U,W;Y),\label{eqnM2}\\
    R_0+R_2 & \leq I(V,W;Z),\label{eqnM3}\\
    R_0+R_1 + R_2 & \leq \min \{I(W;Y),I(W;Z) \} + I(U;Y|W)+ I(V;Z|W) - I(U;V|W),\label{eqnM4}
\end{align}
\end{subequations}
for any triple of random variables $(U,V,W)$ such that $(U,V,W)\mcchain  X\mcchain (Y,Z)$ is achievable for a broadcast channel $T(y,z|x)$. We denote this region as $\mathcal{M}(T)$.
\end{theorem}

Marton's inner bound \cite{mar79} is the best-known achievable region for two-receiver broadcast channels. Several studies have shown that Marton’s inner bound coincides with the capacity region for specific classes of broadcast channels; see  \cite{bergmans73,gallager74,ak75,abbas79,GelPin80,nair10,ggny14,nke16,wss06,gn14}. In several of these cases, Marton's inner bound also coincides with the following outer bound:

\begin{theorem}[UVW-Outer Bound, \cite{nai11arx}]
\label{thm:uvwob}
 If a rate triple $(R_0,R_1,R_2)$ is achievable for a broadcast channel $T(y,z|x)$, then it must satisfy the inequalities
    \begin{subequations}
        \begin{align}
        R_0 & \leq \min \{I(W;Y),I(W;Z) \},\label{uv1}\\ 
        R_0+R_1 & \leq \min \{I(W;Y),I(W;Z) \}+I(U;Y|W),\label{uv2}\\
        R_0+R_2 & \leq \min \{I(W;Y),I(W;Z) \}+I(V;Z|W),\label{uv3}\\
        R_0+R_1 + R_2 & \leq \min \{I(W;Y),I(W;Z) \} + I(U;Y|W)+ I(X;Z|U, W),\label{uv4}\\
        R_0+R_1 + R_2 & \leq \min \{I(W;Y),I(W;Z) \} + I(V;Z|W)+ I(X;Y|V, W),\label{uv5}
    \end{align}
    \end{subequations}
    for some pmf $p(u,v,w,x)$. Further, it suffices to consider $(U,V,W)$ satisfying $|\mathcal{W}| \leq |\mathcal{X}|+5, |\mathcal{U}|\leq |\mathcal{X}|+1, |\mathcal{V}|\leq |\mathcal{X}| + 1$. We denote this region as $\mathcal{O}_{UVW}(T)$.
    \end{theorem}

    \begin{remark}\label{RmkKM}
Removing either the sum-rate constraint \eqref{uv4} or \eqref{uv5} yields the two Körner-Marton outer bounds. \hfill $\diamondsuit$
    \end{remark}
    
    \begin{remark}
The UVW outer bound, like Marton's inner bound,  matches the capacity region for certain classes of broadcast channels, \cite{bergmans73,gallager74,ak75,abbas79,GelPin80,nair10,nke16,wss06,gn14}. In all these cases, $\mathcal{O}_{UVW}(T) = \mathcal{C}(T) = \mathcal{M}(T)$. \hfill $\diamondsuit$
\end{remark}

In \cite{g80}, El Gamal studied product and sum broadcast channels and characterized the capacity region when the component broadcast channels are reversely degraded. For the product and sum broadcast channels considered in \cite{g80}, $\mathcal{O}_{UVW}(T) = \mathcal{C}(T) = \mathcal{M}(T)$. In \cite{ggny14}, a class of product broadcast channels was identified for which Marton's inner bound (Theorem \ref{thm:martonib}) equals the capacity region. Notably, the UVW outer bound yields a strictly larger region, with $\mathcal{O}_{UVW}(T) \supsetneq \mathcal{C}(T) = \mathcal{M}(T)$. The authors demonstrated that Marton's inner bound is tight for both reversely more-capable product broadcast channels and reversely semi-deterministic product broadcast channels. There, the authors derived a new outer bound that relied on the product nature of the underlying broadcast channel.

In Section~\ref{sec:example_reversely-semi}, we provide a sum-broadcast channel with semi‑deterministic components for which $\mathcal{O}_{UVW}(T) \supsetneq \mathcal{M}(T)$. Therefore, the capacity results established in this paper cannot be derived by the standard approach of showing that the $UVW$ outer bound and Marton's inner bound coincide. Moreover, this sum-broadcast channel does not admit a product decomposition, so the outer bound in \cite{ggny14} does not apply.

Recently, the auxiliary receiver approach of \cite{gn22} has proved effective in obtaining sharp capacity results for various problems in network information theory, such as the relay channel and the interference channel (see, e.g., \cite{el2022strengthened,10619605,gn22,wen2024new,chen2025differential}). In this paper, we employ the auxiliary receiver outer bound of \cite{gn22} to establish the capacity region for several non-trivial classes of sum-broadcast channels, including examples where $\mathcal{O}_{UVW}(T) \supsetneq \mathcal{M}(T)$. In particular, we prove that Marton's inner bound is tight for sum broadcast channels with more-capable or semi-deterministic components. 

\section{Preliminaries}
A discrete memoryless broadcast channel (DM-BC) with input alphabet \( \mathcal{X} \) and output alphabets \( \mathcal{Y} \) and \( \mathcal{Z} \) is characterized by the transition probability \( T(y,z | x) \).
 The DM-BC $T(y,z|x)$ is called \emph{semi-deterministic} \cite{mar79} if one of the outputs is a deterministic function of the input; that is, either \(Y = f(X)\) or \(Z = f(X)\) for some deterministic function $f$.  
A DM-BC is said to be \emph{less-noisy} \cite{km77} with $Y$ less-noisy than $Z$ (denoted by \(Y \overset{\mathrm{L.N.}}{\succeq} Z\)) if for every auxiliary random variable $U$ satisfying
$U \mcchain X \mcchain (Y,Z),$
we have
$I(U;Y) \ge I(U;Z)$. 
A DM-BC is said to be \emph{more-capable} \cite{km77} with $Y$ more-capable than $Z$ (denoted by \(Y \overset{\mathrm{M.C.}}{\succeq} Z\)) if for every input distribution $p(x)$, $I(X;Y) \ge I(X;Z)$.

\begin{remark}
Throughout this paper, $\overline\alpha$ denotes $1-\alpha$. All logarithms are in base two. The binary entropy function is defined as
\[H_2(q)=-q\log_2(q)-\overline q\log_2(\overline q), \qquad\forall q\in[0,1].\]
The following equation is useful in our arguments: for any $\lambda>0,\alpha,\beta\geq 0$
\begin{align}
   \max_{q\in [0,1]} \lambda H_2(q) + q \alpha + \overline{q} \beta =  \max_{q\in [0,1]} \lambda \Big(q \log_2 \frac{2^{{\alpha}/{\lambda}}}{q} + \overline{q}  \log_2 \frac{2^{{\beta}/{\lambda}}}{\overline{q}}\Big) = \lambda \log_2\left( 2^{{\alpha}/{\lambda}} +2^{{\beta}/{\lambda}}\right).\label{rem:opt}
\end{align} 
\hfill $\diamondsuit$
\end{remark}

\subsection{Weighted Sum-Rate of the Inner and Outer Bounds}

The (convex) capacity region of a broadcast channel $T_{YZ|X}$ is completely characterized by supporting hyperplanes of the form $\lambda_0 R_0 + \lambda_1 R_1 + \lambda_2 R_2$ for $\lambda_0 \geq \max(\lambda_1, \lambda_2) \geq 0$. This is because if $(R_0,R_1,R_2)$ is achievable, so are $(0, R_0+R_1,R_2)$ and $(0,R_1,R_0+R_2)$. If $\lambda_0<\max(\lambda_1,\lambda_2)$, the entire common rate $R_0$ can be moved to either $R_1$ or $R_2$; we may therefore increase $\lambda_0$ to $\max(\lambda_1,\lambda_2)$ without affecting the maximum weighted sum‑rate.

Take some $\lambda_0 \geq \lambda_1 \geq \lambda_2$. 
By multiplying \eqref{eqnM4} by $\lambda_2$, \eqref{eqnM2} by $(\lambda_1-\lambda_2)$ and \eqref{eqnM1} by $(\lambda_0-\lambda_1)$,  we derive the following inequality for any triple $(R_0,R_1,R_2)$ within Marton’s inner bound:
\begin{align}
   \lambda_0 R_0 + \lambda_1 R_1 + \lambda_2 R_2\leq &
   \max_{p(w,u,v,x)} (\lambda_0-\lambda_1+\lambda_2)\min(I(W;Y),I(W;Z))\label{eqnMm1}\\
    &\qquad\qquad+(\lambda_1-\lambda_2)I(U,W;Y)+\lambda_2(I(U;Y|W)+ I(V;Z|W)-I(U;V|W)).\nonumber
\end{align}
Similarly, for any $\lambda_0\geq \lambda_2\geq \lambda_1$, by multiplying \eqref{eqnM4} by $\lambda_1$, \eqref{eqnM3} by $(\lambda_2-\lambda_1)$ and \eqref{eqnM1} by $(\lambda_0-\lambda_2)$, we obtain
\begin{align}
   \lambda_0 R_0 + \lambda_1 R_1 + \lambda_2 R_2\leq &
   \max_{p(w,u,v,x)} (\lambda_0-\lambda_2+\lambda_1)\min(I(W;Y),I(W;Z))\label{eqnMm2}\\
    &\qquad\qquad+(\lambda_2-\lambda_1)I(V,W;Z)+\lambda_1(I(U;Y|W)+ I(V;Z|W)-I(U;V|W)).\nonumber
\end{align}
\begin{remark}
    In the definitions below, the subscript $1$ corresponds to the case $\lambda_0 \geq \lambda_1 \geq \lambda_2 \geq 0$, while the subscript $2$ corresponds to the case $\lambda_0 \geq \lambda_2 \geq \lambda_1 \geq 0$. \hfill$\diamondsuit$
\end{remark}

The following characterization is known \cite[Eq.(2)]{anantharam2018evaluation}:
\begin{lemma}\label{MartonCharacterization} For any broadcast channel $T(y,z|x)$, we have
    \begin{align*}
        & \max_{(R_0,R_1,R_2)\in \mathcal{M}(T)}\lambda_0 R_0+\lambda_1 R_1+\lambda_2 R_2=\min_{\alpha\in[0,1]} SR_{IB,1}^{\lambda_0,\lambda_1,\lambda_2,\alpha}(T),\quad \forall\lambda_0 \geq \lambda_1 \geq \lambda_2\geq 0,\\
        & \max_{(R_0,R_1,R_2)\in \mathcal{M}(T)}\lambda_0 R_0+\lambda_1 R_1+\lambda_2 R_2=\min_{\alpha\in[0,1]} SR_{IB,2}^{\lambda_0,\lambda_1,\lambda_2,\alpha}(T),\quad \forall\lambda_0 \geq \lambda_2 \geq \lambda_1\geq 0,
    \end{align*}
    where
    \begin{subequations}
        \begin{align}
        \nonumber SR_{IB,1}^{\lambda_0,\lambda_1,\lambda_2,\alpha}(T) :=& \max_{p(w,u,v,x)}\alpha (\lambda_0-\lambda_1+\lambda_2)I(W;Y)+\overline\alpha (\lambda_0-\lambda_1+\lambda_2)I(W;Z)\\
    &+(\lambda_1-\lambda_2)I(U,W;Y)+\lambda_2(I(U;Y|W)+ I(V;Z|W)-I(U;V|W)),\label{defSR_IB,1}\\
        \nonumber SR_{IB,2}^{\lambda_0,\lambda_1,\lambda_2,\alpha}(T) :=& \max_{p(w,u,v,x)} \alpha (\lambda_0-\lambda_2+\lambda_1)I(W;Y)+\overline\alpha (\lambda_0-\lambda_2+\lambda_1)I(W;Z)\\
    &+(\lambda_2-\lambda_1)I(V,W;Z)+\lambda_1(I(U;Y|W)+ I(V;Z|W)-I(U;V|W)).\label{defSR_IB,2}
    \end{align}
    \end{subequations}
\end{lemma}
By an argument analogous to those in \eqref{eqnMm1} and \eqref{eqnMm2}, 
together with the inequality $\max \min \leq \min \max$, 
we obtain an upper bound on the weighted sum-rate of the UVW outer bound.
\begin{align}
        & \max_{(R_0,R_1,R_2)\in \mathcal{O}_{UVW}(T)}\lambda_0 R_0+\lambda_1 R_1+\lambda_2 R_2\leq \min_{\alpha\in[0,1]} SR_{OB,1p}^{\lambda_0,\lambda_1,\lambda_2,\alpha}(T),\quad \forall\lambda_0 \geq \lambda_1 \geq \lambda_2\geq 0,\label{eqnAddl1}\\
        & \max_{(R_0,R_1,R_2)\in \mathcal{O}_{UVW}(T)}\lambda_0 R_0+\lambda_1 R_1+\lambda_2 R_2\leq \min_{\alpha\in[0,1]} SR_{OB,2p}^{\lambda_0,\lambda_1,\lambda_2,\alpha}(T),\quad \forall\lambda_0 \geq \lambda_2 \geq \lambda_1\geq 0,\label{eqnAddl2}
    \end{align}
where
\begin{subequations}
    \begin{align}SR_{OB,1p}^{\lambda_0,\lambda_1,\lambda_2,\alpha}(T):=&\max_{p(w,u,v,x)}\alpha (\lambda_0-\lambda_1+\lambda_2)I(W;Y)+\overline\alpha (\lambda_0-\lambda_1+\lambda_2)I(W;Z)+(\lambda_1-\lambda_2)I(U,W;Y)\nonumber\\
    &+\lambda_2\min\{I(U;Y|W)+I(X;Z|U,W),~I(V;Z|W)+I(X;Y|V,W)\},\label{defSR_OB1p}\\
    SR_{OB,2p}^{\lambda_0,\lambda_1,\lambda_2,\alpha}(T):=&\max_{p(w,u,v,x)}\alpha (\lambda_0-\lambda_2+\lambda_1)I(W;Y)+\overline\alpha (\lambda_0-\lambda_2+\lambda_1)I(W;Z)+(\lambda_2-\lambda_1)I(V,W;Z)\nonumber\\
    &+\lambda_1\min\{I(U;Y|W)+I(X;Z|U,W),~I(V;Z|W)+I(X;Y|V,W)\}.\label{defSR_OB2p}
    \end{align}
\end{subequations}
We also define the following expressions which correspond to the weighted sum-rate of the Körner-Marton outer bound (see Remark \ref{RmkKM}):
\begin{align*}
SR_{OB,1m}^{\lambda_0,\lambda_1,\lambda_2,\alpha}(T) :=&\max_{p(w,u,x)}\alpha (\lambda_0-\lambda_1+\lambda_2)I(W;Y)+\overline\alpha (\lambda_0-\lambda_1+\lambda_2)I(W;Z)+(\lambda_1-\lambda_2)I(U,W;Y)\\
    &+\lambda_2(I(U;Y|W)+ I(X;Z|U,W)),\\
        SR_{OB,1n}^{\lambda_0,\lambda_1,\lambda_2,\alpha}(T) :=&\max_{p(w,u,v,x)}\alpha (\lambda_0-\lambda_1+\lambda_2)I(W;Y)+\overline\alpha (\lambda_0-\lambda_1+\lambda_2)I(W;Z)+(\lambda_1-\lambda_2)I(U,W;Y)\\
    &+\lambda_2(I(V;Z|W)+ I(X;Y|V,W)), \\
    \stackrel{(a)}{=} & \max_{p(w,v,x)}\alpha (\lambda_0-\lambda_1+\lambda_2)I(W;Y)+\overline\alpha (\lambda_0-\lambda_1+\lambda_2)I(W;Z)+(\lambda_1-\lambda_2)I(X;Y)\\
    &+\lambda_2(I(V;Z|W)+ I(X;Y|V,W)), \\
SR_{OB,2m}^{\lambda_0,\lambda_1,\lambda_2,\alpha}(T) :=&\max_{p(w,v,x)}\alpha (\lambda_0-\lambda_2+\lambda_1)I(W;Y)+\overline\alpha (\lambda_0-\lambda_2+\lambda_1)I(W;Z)+(\lambda_2-\lambda_1)I(V,W;Z)\\
    &+\lambda_1(I(V;Z|W)+ I(X;Y|V,W)),\\
        SR_{OB,2n}^{\lambda_0,\lambda_1,\lambda_2,\alpha}(T) :=&\max_{p(w,u,v,x)}\alpha (\lambda_0-\lambda_2+\lambda_1)I(W;Y)+\overline\alpha (\lambda_0-\lambda_2+\lambda_1)I(W;Z)+(\lambda_2-\lambda_1)I(V,W;Z)\\
    &+\lambda_1(I(U;Y|W)+ I(X;Z|U,W)),\\
    \stackrel{(b)}{=}&\max_{p(w,u,x)}\alpha (\lambda_0-\lambda_2+\lambda_1)I(W;Y)+\overline\alpha (\lambda_0-\lambda_2+\lambda_1)I(W;Z)+(\lambda_2-\lambda_1)I(X;Z)\\
    &+\lambda_1(I(U;Y|W)+ I(X;Z|U,W)).
\end{align*}
Here, $(a)$ follows from the data-processing inequality $I(U,W;Y) \leq I(X;Y)$, so that setting $U=X$ is optimal; an analogous argument applies for $(b)$.

\begin{lemma}\label{lem:srobcompare}
For any channel $T$ and any $\alpha\in[0,1]$:
\begin{subequations}
\begin{align}
&SR_{IB,1}^{\lambda_0,\lambda_1,\lambda_2,\alpha}(T)\leq SR_{OB,1p}^{\lambda_0,\lambda_1,\lambda_2,\alpha}(T)\leq 
    \min\left\{SR_{OB,1m}^{\lambda_0,\lambda_1,\lambda_2,\alpha}(T),SR_{OB,1n}^{\lambda_0,\lambda_1,\lambda_2,\alpha}(T)\right\},\quad \forall\lambda_0 \geq \lambda_1 \geq \lambda_2\geq 0,\label{eq:con7}\\
    &SR_{IB,2}^{\lambda_0,\lambda_1,\lambda_2,\alpha}(T)\leq SR_{OB,2p}^{\lambda_0,\lambda_1,\lambda_2,\alpha}(T)\leq 
    \min\left\{SR_{OB,2m}^{\lambda_0,\lambda_1,\lambda_2,\alpha}(T),SR_{OB,2n}^{\lambda_0,\lambda_1,\lambda_2,\alpha}(T)\right\},\quad \forall\lambda_0 \geq \lambda_2 \geq \lambda_1\geq 0.\label{eq:con8}
\end{align}
\end{subequations}
\end{lemma}
The lemma is immediate from the fact that $\mathcal{M}(T)\subseteq \mathcal{O}_{UVW}(T)$ and the  max–min inequality ($\max\min \leq \min\max$).

\subsection{Primary Broadcast Channels}
\label{sec:primary}

\begin{definition}
\label{defn:primary}
For a given triple
$\lambda_0 \geq \lambda_1 \geq \lambda_2 \geq 0$, we say that a broadcast channel $T_{YZ|X}$ belongs to the \emph{primary class}
$\mathcal{P}^{\lambda_0,\lambda_1,\lambda_2}$ if, for every $\alpha\in[0,1]$,
\begin{subequations}
\begin{align}
SR_{IB,1}^{\lambda_0,\lambda_1,\lambda_2,\alpha}(T)
\;\geq\;
\min\!\big\{
SR_{OB,1m}^{\lambda_0,\lambda_1,\lambda_2,\alpha}(T),
SR_{OB,1n}^{\lambda_0,\lambda_1,\lambda_2,\alpha}(T)
\big\}.
\label{eq:con1}
\end{align}
Similarly, for a given triple $\lambda_0 \geq \lambda_2 \geq \lambda_1 \geq 0$,
we say that $T_{YZ|X}$ belongs to the primary class
$\mathcal{P}^{\lambda_0,\lambda_1,\lambda_2}$ if, for every $\alpha\in[0,1]$,
\begin{align}
SR_{IB,2}^{\lambda_0,\lambda_1,\lambda_2,\alpha}(T)
\;\geq\;
\min\!\big\{
SR_{OB,2m}^{\lambda_0,\lambda_1,\lambda_2,\alpha}(T),
SR_{OB,2n}^{\lambda_0,\lambda_1,\lambda_2,\alpha}(T)
\big\}.
\label{eq:con2}
\end{align}
\end{subequations}
\end{definition}

\begin{remark}
The ordering of $(\lambda_0,\lambda_1,\lambda_2)$ dictates whether condition \eqref{eq:con1} or \eqref{eq:con2}  applies, and thus identifies the appropriate instance of $\mathcal{P}^{\lambda_0,\lambda_1,\lambda_2}$. \hfill$\diamondsuit$
\end{remark}

\begin{lemma}\label{lem:primaexam}  
    If $T_{YZ|X}$ is either a less-noisy or semi-deterministic broadcast channel, then $T_{YZ|X} \in \mathcal{P}^{\lambda_0,\lambda_1,\lambda_2}$ for all $\lambda_0\geq \max\{\lambda_1,\lambda_2\}\geq 0$. 
\end{lemma}
\begin{proof}Please refer to Appendix~\ref{specialtogoodproof}.\end{proof}

The class $\mathcal{P}^{\lambda_0,\lambda_1,\lambda_2}$ encompasses a broader range of channels than just the less-noisy or semi-deterministic types. In particular, it is shown in Appendix~\ref{sec:appndAF} that MIMO Gaussian channels belong to $\mathcal{P}^{\lambda_0,\lambda_1,\lambda_2}$. On the other hand, $\mathcal{P}^{\lambda_0,\lambda_1,\lambda_2}$ does not, in general, contain more-capable broadcast channels; see Appendix~\ref{sec:appndAE}. To address this limitation, we introduce an alternative primary class.
\begin{definition}\label{defn:primary2}
    For a given triple
$\lambda_0 \geq \max\{\lambda_1, \lambda_2\} \geq 0$, we say that a broadcast channel $T_{YZ|X}$ belongs to the primary class $\hat{\mathcal{P}}^{\lambda_0,\lambda_1,\lambda_2}$ if $T\in \hat{\mathcal{P}}^{\lambda_0,\lambda_1,\lambda_2}_{A}\cup\hat{\mathcal{P}}^{\lambda_0,\lambda_1,\lambda_2}_{B}$, as defined below:
    \begin{itemize}
        \item For a given triple
$\lambda_0 \geq \lambda_1 \geq \lambda_2 \geq 0$, we say that  $T_{YZ|X}$ belongs to the \textit{primary class} $\hat{\mathcal{P}}^{\lambda_0,\lambda_1,\lambda_2}_{A}$ if, for every $\alpha\in[0,1]$,
    \begin{subequations}
    \begin{align}
    SR_{IB,1}^{\lambda_0,\lambda_1,\lambda_2,\alpha}(T)\geq & SR_{OB,1p}^{\lambda_0,\lambda_1,\lambda_2,\alpha}(T). \label{eq:con3}
    \end{align}
    Correspondingly, for a given triple
$\lambda_0 \geq \lambda_2 \geq \lambda_1 \geq 0$, we say that $T_{YZ|X}$ belongs to the \textit{primary class} $\hat{\mathcal{P}}^{\lambda_0,\lambda_1,\lambda_2}_{A}$ if, for every $\alpha\in[0,1]$,
    \begin{align}
    SR_{IB, 2}^{\lambda_0,\lambda_1,\lambda_2,\alpha}(T)\geq &SR_{OB,2m}^{\lambda_0,\lambda_1,\lambda_2,\alpha}(T). \label{eq:con4}
    \end{align}
    \end{subequations}
    \item For a given triple
$\lambda_0 \geq \lambda_1 \geq \lambda_2 \geq 0$, we say that $T_{YZ|X}$ belongs to the \textit{primary class} $\hat{\mathcal{P}}^{\lambda_0,\lambda_1,\lambda_2}_{B}$ if, for every $\alpha\in[0,1]$,
    \begin{subequations}
    \begin{align}
    SR_{IB, 1}^{\lambda_0,\lambda_1,\lambda_2,\alpha}(T)\geq &SR_{OB,1m}^{\lambda_0,\lambda_1,\lambda_2,\alpha}(T). \label{eq:con5}
    \end{align}
    Correspondingly, for a given triple
$\lambda_0 \geq \lambda_2 \geq \lambda_1 \geq 0$, we say that $T_{YZ|X}$ belongs to the \textit{primary class} $\hat{\mathcal{P}}^{\lambda_0,\lambda_1,\lambda_2}_{B}$ if, for every $\alpha\in[0,1]$,
    \begin{align}
    SR_{IB, 2}^{\lambda_0,\lambda_1,\lambda_2,\alpha}(T)\geq &SR_{OB,2p}^{\lambda_0,\lambda_1,\lambda_2,\alpha}(T).\label{eq:con6}
    \end{align}
    \end{subequations}
    \end{itemize}
\end{definition}
\begin{remark}
    Observe that
    $$\hat{\mathcal{P}}^{\lambda_0,\lambda_1,\lambda_2}_{B}\subseteq \mathcal{P}^{\lambda_0,\lambda_1,\lambda_2}\subseteq \hat{\mathcal{P}}^{\lambda_0,\lambda_1,\lambda_2}_{A}, \qquad\forall \lambda_0 \geq \lambda_1 \geq \lambda_2 \geq 0,$$
$$\hat{\mathcal{P}}^{\lambda_0,\lambda_1,\lambda_2}_{A}\subseteq \mathcal{P}^{\lambda_0,\lambda_1,\lambda_2}\subseteq \hat{\mathcal{P}}^{\lambda_0,\lambda_1,\lambda_2}_{B}, \qquad\forall \lambda_0 \geq \lambda_2 \geq \lambda_1 \geq 0,$$
 \hfill$\diamondsuit$
\end{remark}
\begin{lemma}\label{lem:primaexam2}
    If $T_{YZ|X}$ is a more-capable broadcast channel, then $T_{YZ|X} \in \hat{\mathcal{P}}^{\lambda_0,\lambda_1,\lambda_2},\forall \lambda_0\geq \max\{\lambda_1,\lambda_2\}\geq 0$. Specifically, when receiver $Z$ is more-capable than $Y$, then $T_{YZ|X} \in \hat{\mathcal{P}}^{\lambda_0,\lambda_1,\lambda_2}_{A},\forall \lambda_0\geq \max\{\lambda_1,\lambda_2\}\geq 0$. Conversely, when receiver $Y$ is more-capable than $Z$, then $T_{YZ|X} \in \hat{\mathcal{P}}^{\lambda_0,\lambda_1,\lambda_2}_{B},\forall \lambda_0\geq \max\{\lambda_1,\lambda_2\}\geq 0$. Moreover, any less-noisy or semi-deterministic broadcast channel belongs to $\hat{\mathcal{P}}^{\lambda_0,\lambda_1,\lambda_2},\forall \lambda_0\geq \max\{\lambda_1,\lambda_2\}\geq 0$.
\end{lemma}
\begin{proof}
Please refer to Appendix~\ref{specialtogoodproof}. \end{proof}

The following lemma is an immediate consequence of our definitions.
\begin{lemma}\label{lem:primcap}
If a broadcast channel $T$ belongs to $\mathcal{P}^{\lambda_0,\lambda_1,\lambda_2}$ or $\hat{\mathcal{P}}^{\lambda_0,\lambda_1,\lambda_2}$ for some weights $\lambda_0\geq \max\{\lambda_1,\lambda_2\}\geq 0$,
then the maximum weighted sum‑rate over Marton’s inner bound (Theorem \ref{thm:martonib}) equals that over the UVW outer bound (Theorem \ref{thm:uvwob}), thus giving the capacity region’s weighted sum‑rate:
\begin{align*}
        &\max_{(R_0,R_1,R_2)\in \mathcal{O}_{UVW}(T)}\lambda_0 R_0+\lambda_1 R_1+\lambda_2 R_2 = \max_{(R_0,R_1,R_2)\in \mathcal{M}(T)}\lambda_0 R_0+\lambda_1 R_1+\lambda_2 R_2.
    \end{align*}
In particular, if $T \in \mathcal{P}^{\lambda_0,\lambda_1,\lambda_2}\cup \hat{\mathcal{P}}^{\lambda_0,\lambda_1,\lambda_2}$ for \underline{all} $\lambda_0\geq \max\{\lambda_1,\lambda_2\}\geq 0$, then the entire Marton's inner bound $\mathcal{M}(T)$ (Theorem \ref{thm:martonib}) and the UVW outer bound $\mathcal{O}_{UVW}(T)$ (Theorem \ref{thm:uvwob}) coincide, yielding the capacity region.
\end{lemma}
\begin{proof}

Let $\lambda_0 \geq \lambda_1 \geq \lambda_2 \geq 0$ be given.     Suppose $T\in\hat{\mathcal{P}}^{\lambda_0,\lambda_1,\lambda_2}_{A}$. We have
    \begin{align*}
        &\max_{(R_0,R_1,R_2)\in \mathcal{O}_{UVW}(T)}\lambda_0 R_0+\lambda_1 R_1+\lambda_2 R_2\\
        &\quad \leq \max_{p(w,u,v,x)} (\lambda_0-\lambda_1+\lambda_2)\min\{I(W;Y), I(W;Z)\}+(\lambda_1-\lambda_2) I(U,W;Y)\\
        &\qquad +\lambda_2\min\{I(U;Y|W)+I(X;Z|U,W),I(V;Z|W)+I(X;Y|V,W)\}\\
        &\quad \leq \min_{\alpha\in[0,1]}SR_{OB,1p}^{\lambda_0,\lambda_1,\lambda_2,\alpha}(T)\\
        &\quad \stackrel{(a)}\leq \min_{\alpha\in[0,1]}SR_{IB,1}^{\lambda_0,\lambda_1,\lambda_2,\alpha}(T)\\
        &\quad \stackrel{(b)}= \max_{(R_0,R_1,R_2)\in \mathcal{M}(T)}\lambda_0 R_0+\lambda_1 R_1+\lambda_2 R_2,
    \end{align*}
    where $(a)$ follows from the assumption, and $(b)$ follows from Lemma~\ref{MartonCharacterization}. The proofs for the remaining cases are analogous and omitted.
\end{proof}


\definecolor{MartonColor}{RGB}{245,245,220}   
\definecolor{Aonly}{RGB}{174,199,232}         
\definecolor{Bonly}{RGB}{255,187,120}         
\definecolor{Conly}{RGB}{141,211,199}         
\definecolor{ABonly}{RGB}{197,176,213}        
\definecolor{AConly}{RGB}{255,152,150}        
\definecolor{BConly}{RGB}{178,150,210}        
\definecolor{ABC}{RGB}{127,201,127}           

\begin{figure}[t]
\centering

\begin{tikzpicture}[scale=0.95, every node/.style={font=\small}]

    \def\Outer{(0,0) circle (4.2)}
    \def\A{(-1.05,0) ellipse [x radius=2.05, y radius=1.25]}
    \def\B{( 1.05,0) ellipse [x radius=2.05, y radius=1.25]}
    \def\C{(0,0) ellipse [x radius=1.15, y radius=1.55]}
    
    \def\Box{(-4.5,-4.5) rectangle (4.5,4.5)}

    \fill[MartonColor] \Outer;

    \begin{scope}
        \clip \Outer; \clip \A; \fill[Aonly] \Box;
    \end{scope}

    \begin{scope}
        \clip \Outer; \clip \B; \fill[Bonly] \Box;
    \end{scope}

    \begin{scope}
        \clip \Outer; \clip \C; \fill[Conly] \Box;
    \end{scope}

    \begin{scope}
        \clip \Outer; \clip \A; \clip \B; \fill[ABonly] \Box;
    \end{scope}

    \begin{scope}
        \clip \Outer; \clip \A; \clip \C; \fill[AConly] \Box;
    \end{scope}

    \begin{scope}
        \clip \Outer; \clip \B; \clip \C; \fill[BConly] \Box;
    \end{scope}

    \begin{scope}
        \clip \Outer; \clip \A; \clip \B; \clip \C; \fill[ABC] \Box;
    \end{scope}

    \draw[thick, draw=black!80] \Outer;
    \draw[thick, blue!60!black] \A;
    \draw[thick, orange!80!black] \B;
    \draw[thick, green!50!black] \C;

    \draw[thick, black!80, rounded corners=15pt] (-4.5,-4.5) rectangle (4.5,4.5);

    \tikzstyle{labelbox} = [inner sep=2pt, rounded corners=3pt]
    \node[labelbox] at (-2.8,1.10) {\(\hat{\mathcal{P}}^{\lambda_0,\lambda_1,\lambda_2}_{A}\)};
    \node[labelbox] at ( 3,1.10) {\(\hat{\mathcal{P}}^{\lambda_0,\lambda_1,\lambda_2}_{B}\)};
    \node[labelbox] at (0,-1.8) {\(\mathcal{P}^{\lambda_0,\lambda_1,\lambda_2}\)};

\end{tikzpicture}
\caption{Venn diagram of broadcast channel classes.}

\vspace{2em} 

\newcommand{\legendbox}[1]{\tikz[baseline=-0.1ex]{\draw[gray!50, fill=#1] (0,0) rectangle (0.4,0.25);}}

\renewcommand{\arraystretch}{1.4} 
\begin{tabular}{ll} 
\hline
\textbf{Class of Broadcast Channel} & \textbf{Color of the region it belongs to} \\
\hline
More-capable broadcast channels: \(Z \overset{\mathrm{M.C.}}{\succeq} Y\) & \legendbox{Aonly} \(+\) \legendbox{AConly} \(+\) \legendbox{ABC} \\
More-capable broadcast channels: \(Y \overset{\mathrm{M.C.}}{\succeq} Z\) & \legendbox{Bonly} \(+\) \legendbox{BConly} \(+\) \legendbox{ABC} \\
Semi-deterministic broadcast channels: \(Y=f(X)\) & \legendbox{AConly} \(+\) \legendbox{ABC} \\
Semi-deterministic broadcast channels: \(Z=f(X)\) & \legendbox{BConly} \(+\) \legendbox{ABC} \\
Less-noisy broadcast channels: \(Y \overset{\mathrm{L.N.}}{\succeq} Z\) or \(Z \overset{\mathrm{L.N.}}{\succeq} Y\) & \legendbox{ABC} \\
Broadcast channels for which Marton’s inner bound is tight & \legendbox{Aonly} \(+\) \legendbox{AConly} \(+\) \legendbox{Bonly} \(+\) \legendbox{BConly} \(+\) \legendbox{ABC} \(+\) \legendbox{Conly} \(+\) \legendbox{MartonColor} \\
\hline
\end{tabular}
\end{figure}
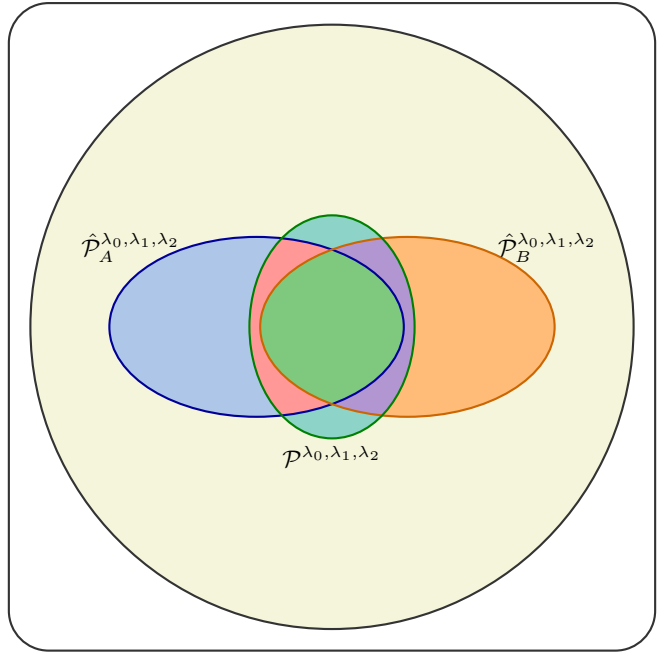

\subsection{Auxiliary Receiver Outer Bound}
\begin{theorem}[Theorem 8, \cite{gn22}] \label{th:obp} Consider an arbitrary broadcast channel \(T_{Y,Z | X}\). 
Let \(\mathcal{G}\) and \(\mathcal{K}\) be arbitrary sets, and let 
\(T_{G,K | X,Y,Z}\) denote a transition probability, where 
\(G\in \mathcal{G}\) and \(K\in \mathcal{K}\) are auxiliary receivers. Then, any achievable non-negative rate triple $(R_0, R_1, R_2)$ must satisfy the following constraints
{\begin{subequations} \begin{align}
R_0 & \leq \min \{ I(W^\dagger; G) + I(W; Y|G), I(W^\dagger; Z|K) + I(W; K) \}, \label{eqnSUM1a}\\
R_0 + R_1 & \leq I(U^\dagger,W^\dagger; G) + I(U,W; Y|G), \label{eqnSUM2a} \\
R_0 + R_1 & \leq I(W^\dagger; Z|K) + I(W, G; K)+I(U^\dagger; G|W^\dagger, K) + I(U; Y|W, G),  \label{eqnSUM3a}\\
R_0 + R_2 & \leq I(W^\dagger, K; G)+I(W; Y|G) +I(V^\dagger; Z|W^\dagger, K)+I(V; K|W, G), \label{eqnSUM4a}\\
R_0 + R_2 & \leq 
I(V,W; K) + I(V^\dagger,W^\dagger; Z|K), \label{eqnSUM5a}\\
R_0 + R_1 + R_2 & \leq \min \{ I(W^\dagger, K; G) + I(W; Y|G),~ I(W^\dagger; Z|K) + I(W, G; K) \} \nonumber \\
&\qquad+ I(U;Y|W, G) + I(X;K|U, W, G)\nonumber \\
& \qquad + \min \big\{ I(U^\dagger;G|W^\dagger, K) + I(X; Z|U^\dagger, W^\dagger, K),~ I(V^\dagger;Z|W^\dagger, K) + I(X; G|V^\dagger, W^\dagger, K) \big\},\label{eqnSUM6a} \\
R_0 + R_1 + R_2 & \leq \min \{ I(W^\dagger, K; G) + I(W; Y|G),~ I(W^\dagger; Z|K) + I(W, G; K) \} \nonumber\\
&\qquad + I(V^\dagger;Z|W^\dagger, K) + I(X;G|V^\dagger, W^\dagger, K) \nonumber\\&\qquad+ \min \big\{ I(U;Y|W, G) + I(X; K|U, W, G),~I(V;K|W, G) + I(X; Y|V, W, G) \big\},
\label{eqnSUM7a} \end{align}\end{subequations}}
 for some $p(w, v, u|x)p(w^\dagger, v^\dagger, u^\dagger|x)p(x)$ satisfying $|\mathcal{W}^\dagger|, |\mathcal{W}|\leq |\mathcal{X}|+7$, $|\mathcal{U}^\dagger|, |\mathcal{V}|\leq |\mathcal{X}|+2$, $|\mathcal{V}^\dagger|, |\mathcal{U}|\leq |\mathcal{X}|+1$. We denote this region by $\mathcal{O}_{aux}(T)$.
\end{theorem}

\begin{remark}
    \label{rem:con}
    In Theorem \ref{th:obp}, let us set $G=Y$ and let $K$ be a constant random variable. The constraints then become:
    \begin{align*}
R_0 & \leq \min \{ I(W^\dagger; Y) , I(W^\dagger; Z)  \}, \\
R_0 + R_1 & \leq I(U^\dagger,W^\dagger; Y) ,  \\
R_0 + R_1 & \leq I(W^\dagger; Z) +I(U^\dagger; Y|W^\dagger),  \\
R_0 + R_2 & \leq I(W^\dagger ; Y) +I(V^\dagger; Z|W^\dagger), \\
R_0 + R_2 & \leq I(V^\dagger,W^\dagger; Z), \\
R_0 + R_1 + R_2 & \leq \min \{ I(W^\dagger; Y) ,~ I(W^\dagger; Z)  \} + \min \big\{ I(U^\dagger;Y|W^\dagger) + I(X; Z|U^\dagger, W^\dagger),~ I(V^\dagger;Z|W^\dagger) + I(X; Y|V^\dagger, W^\dagger) \big\}, \\
R_0 + R_1 + R_2 & \leq \min \{ I(W^\dagger; Y) ,~ I(W^\dagger; Z)  \}  + I(V^\dagger;Z|W^\dagger) + I(X;Y|V^\dagger, W^\dagger).
\end{align*}
A moment's reflection reveals that these constraints are equivalent to those in the UVW outer bound. Given that $\mathcal{O}_{aux}(T)$ is obtained by intersecting the resulting regions over all possible choices of $T_{G,K|X,Y,Z}$, it is immediate that $\mathcal{O}_{aux}(T) \subseteq \mathcal{O}_{UVW}(T)$. \hfill $\diamondsuit$
\end{remark}

\subsection{Weighted sum-rate of Marton's inner bound for Sum-Broadcast channels}

Given two broadcast channels $T_a$ and $T_b$, whose input and output sets satisfy $\mathcal{X}_a \cap \mathcal{X}_b = \emptyset$, $\mathcal{Y}_a \cap \mathcal{Y}_b = \emptyset$, $\mathcal{Z}_a \cap \mathcal{Z}_b = \emptyset$,  we make the following definition.
\begin{definition}
    A broadcast channel \(T\) is said to be the sum of two broadcast channels \(T_a\) and \(T_b\), denoted by 
\[
T = T_a \oplus T_b,
\]
if its input and output alphabets are defined as
\[
\mathcal{X} = \mathcal{X}_a \sqcup \mathcal{X}_b, 
\qquad 
\mathcal{Y} = \mathcal{Y}_a \sqcup \mathcal{Y}_b, 
\qquad 
\mathcal{Z} = \mathcal{Z}_a \sqcup \mathcal{Z}_b,
\]
where \(\sqcup\) denotes the disjoint union of two sets; and the channel transition probabilities are defined by
\[
T(y,z | x) =
\begin{cases}
T_a(y,z | x), & \text{if } (x,y,z) \in \mathcal{X}_a \times \mathcal{Y}_a \times \mathcal{Z}_a, \\[6pt]
T_b(y,z | x), & \text{if } (x,y,z) \in \mathcal{X}_b \times \mathcal{Y}_b \times \mathcal{Z}_b, \\[6pt]
0, & \text{otherwise}.
\end{cases}
\]
\end{definition}

Here, we establish some general results regarding the weighted sum-rate of Marton's inner bound for a sum-broadcast channel. 
\begin{lemma}
    Let  $T=T_a\oplus T_b$ be the sum of two broadcast channels $T_a$ and $T_b$. Then for any $\alpha\in[0,1]$, we have
    \begin{subequations}
        \begin{align}
            SR_{IB, 1}^{\lambda_0,\lambda_1,\lambda_2,\alpha}(T)& =  \max_{q\in[0,1]} \lambda_0 H_2(q) + q SR_{IB,1}^{\lambda_0,\lambda_1,\lambda_2,\alpha}(T_a) + \overline{q} SR_{IB,1}^{\lambda_0,\lambda_1,\lambda_2,\alpha}(T_b),\quad \forall\lambda_0 \geq \lambda_1 \geq \lambda_2\geq 0,\label{eqn:martonsumcharacterization1}\\
            SR_{IB, 2}^{\lambda_0,\lambda_1,\lambda_2,\alpha}(T)& =  \max_{q\in[0,1]} \lambda_0 H_2(q) + q SR_{IB,2}^{\lambda_0,\lambda_1,\lambda_2,\alpha}(T_a) + \overline{q} SR_{IB,2}^{\lambda_0,\lambda_1,\lambda_2,\alpha}(T_b),\quad \forall\lambda_0 \geq \lambda_2 \geq \lambda_1\geq 0.\label{eqn:martonsumcharacterization2}
        \end{align}
    \end{subequations}
\end{lemma}
\begin{proof}
    We prove \eqref{eqn:martonsumcharacterization1}. We first claim that there exist maximizers of 
\( SR_{IB,1}^{\lambda_0,\lambda_1,\lambda_2,\alpha}(T) \) 
of the form 
\( ((W,Q), U, V, (X_Q, Q)) \) 
such that \( Q \) is a function of \( W \). To see this,  
    suppose $(W^*,U^*,V^*,X^*)$ is an optimal random variable tuple for the following maximization problem, which defines $SR_{IB,1}^{\lambda_0,\lambda_1,\lambda_2,\alpha}(T)$:
\begin{align*}&\max_{p(w,u,v,x)}\alpha (\lambda_0-\lambda_1+\lambda_2)I(W;Y)+\overline\alpha (\lambda_0-\lambda_1+\lambda_2)I(W;Z)\\
        &\qquad +(\lambda_1-\lambda_2)I(U,W;Y)+\lambda_2(I(U;Y|W)+ I(V;Z|W)-I(U;V|W)).
    \end{align*}
    
    Let $Q^*\in\{a,b\}$ be such that $p(Q^*=a) = p(X^*\in\mathcal{X}_a)=q^*$, and note we can write $X^* = (X_Q^*,Q^*)$. We now show that $(\tilde W, U^*, V^*,X^*)$ is also an optimal tuple for the same maximization problem, where $\tilde W = (W^*,Q^*)$:
    \begin{align*}
        &\alpha (\lambda_0-\lambda_1+\lambda_2)I(\tilde W;Y)+\overline\alpha (\lambda_0-\lambda_1+\lambda_2)I(\tilde W;Z)+(\lambda_1-\lambda_2)I(U^*,\tilde W;Y)\\
        &+\lambda_2(I(U^*;Y|\tilde W)+ I(V^*;Z|\tilde W)-I(U^*;V^*|\tilde W))\\
        =&\alpha (\lambda_0-\lambda_1+\lambda_2)I( W^*;Y)+\overline\alpha (\lambda_0-\lambda_1+\lambda_2)I(W^*;Z)+(\lambda_1-\lambda_2)I(U^*,W^*;Y)\\
        &+\lambda_2(I(U^*;Y,Q^*|W^*)+ I(V^*;Z,Q^*|W^*)-I(U^*;V^*|W^*))\\
        &+\alpha (\lambda_0-\lambda_1+\lambda_2)I(Q^*;Y|W^*)+\overline\alpha (\lambda_0-\lambda_1+\lambda_2)I(Q^*;Z|W^*)+(\lambda_1-\lambda_2)I(Q^*;Y|U^*,W^*)\\
        &+\lambda_2(-I(U^*;Q^*|W^*)-I(V^*;Q^*|W^*)+I(U^*;V^*|W^*)-I(U^*;V^*|W^*,Q^*))\\
        \stackrel{(a)}=&\alpha (\lambda_0-\lambda_1+\lambda_2)I( W^*;Y)+\overline\alpha (\lambda_0-\lambda_1+\lambda_2)I(W^*;Z)+(\lambda_1-\lambda_2)I(U^*,W^*;Y)\\
        &+\lambda_2(I(U^*;Y|W^*)+ I(V^*;Z|W^*)-I(U^*;V^*|W^*))\\
        &+(\lambda_0-\lambda_1)H(Q^*|W^*)+(\lambda_1-\lambda_2)H(Q^*|U^*,W^*)+\lambda_2 H(Q^*|U^*,V^*,W^*)\\
        \geq &\alpha (\lambda_0-\lambda_1+\lambda_2)I( W^*;Y)+\overline\alpha (\lambda_0-\lambda_1+\lambda_2)I(W^*;Z)+(\lambda_1-\lambda_2)I(U^*,W^*;Y)\\
        &+\lambda_2(I(U^*;Y|W^*)+ I(V^*;Z|W^*)-I(U^*;V^*|W^*)),
    \end{align*}
    where $(a)$ follows because $Q^*$ is a function of $Y$ or $Z$.
    
    The above argument establishes our claim that maximizers of the form $((W,Q),U,V,(X_Q,Q))$ exist, with $Q$ a function of $W$.    
    By this property of the maximizers, we obtain:
    \begin{align*}
        &SR_{IB, 1}^{\lambda_0,\lambda_1,\lambda_2,\alpha}(T)
        \\
        &\quad =\max_{p(w,u,v,x_q,q)}\alpha (\lambda_0-\lambda_1+\lambda_2)I(W,Q;Y)+\overline\alpha (\lambda_0-\lambda_1+\lambda_2)I(W,Q;Z)\\
        &\qquad +(\lambda_1-\lambda_2)I(U,W,Q;Y)+\lambda_2(I(U;Y|W,Q)+ I(V;Z|W,Q)-I(U;V|W,Q))\\
        &\quad =\max_{p(w,u,v,x_q,q)} \lambda_0 H(Q)+\alpha (\lambda_0-\lambda_1+\lambda_2)I(W;Y|Q)+\overline\alpha (\lambda_0-\lambda_1+\lambda_2)I(W;Z|Q)\\
        &\qquad +(\lambda_1-\lambda_2)I(U,W;Y|Q)+\lambda_2(I(U;Y|W,Q)+ I(V;Z|W,Q)-I(U;V|W,Q))\\
        &\quad =\max_{q} \lambda_0 H_2(q)+q\bigg(\max_{p(w,u,v,x_a|q=a)}\alpha (\lambda_0-\lambda_1+\lambda_2)I(W;Y_a|Q=a)+\overline\alpha (\lambda_0-\lambda_1+\lambda_2)I(W;Z_a|Q=a)\\
        &\qquad +(\lambda_1-\lambda_2)I(U,W;Y_a|Q=a)+\lambda_2(I(U;Y_a|W,Q=a)+ I(V;Z_a|W,Q=a)-I(U;V|W,Q=a))\bigg)\\
        &\qquad +\overline q\bigg(\max_{p(w,u,v,x_b|q=b)}\alpha (\lambda_0-\lambda_1+\lambda_2)I(W;Y_b|Q=b)+\overline\alpha (\lambda_0-\lambda_1+\lambda_2)I(W;Z_b|Q=b)\\
        &\qquad +(\lambda_1-\lambda_2)I(U,W;Y_b|Q=b)+\lambda_2(I(U;Y_b|W,Q=b)+ I(V;Z_b|W,Q=b)-I(U;V|W,Q=b))\bigg)\\
        &\quad =\max_{q} \lambda_0 H_2(q)+q\bigg(\max_{p(w_a,u_a,v_a,x_a)}\alpha (\lambda_0-\lambda_1+\lambda_2)I(W_a;Y_a)+\overline\alpha (\lambda_0-\lambda_1+\lambda_2)I(W_a;Z_a)\\
        &\qquad +(\lambda_1-\lambda_2)I(U_a,W_a;Y_a)+\lambda_2(I(U_a;Y_a|W_a)+ I(V_a;Z_a|W_a)-I(U_a;V_a|W_a))\bigg)\\
        &\qquad+\overline q\bigg(\max_{p(w_b,u_b,v_b,x_b)}\alpha (\lambda_0-\lambda_1+\lambda_2)I(W_b;Y_b)+\overline\alpha (\lambda_0-\lambda_1+\lambda_2)I(W_b;Z_b)\\
        &\qquad +(\lambda_1-\lambda_2)I(U_b,W_b;Y_b)+\lambda_2(I(U_b;Y_b|W_b)+ I(V_b;Z_b|W_b)-I(U_b;V_b|W_b))\bigg)\\
        &\quad =\max_{q} \lambda_0 H_2(q) + q SR_{IB,1}^{\lambda_0,\lambda_1,\lambda_2,\alpha}(T_a) + \overline{q} SR_{IB,1}^{\lambda_0,\lambda_1,\lambda_2,\alpha}(T_b).
    \end{align*}
    This proves \eqref{eqn:martonsumcharacterization1}. The identity \eqref{eqn:martonsumcharacterization2} can be shown by an entirely analogous argument.
\end{proof}
\begin{corollary}
\label{MartonCharacterizationSumChannel} Let  $T=T_a\oplus T_b$ be the sum of two broadcast channels $T_a$ and $T_b$. Then, for any $\lambda_0 \geq \lambda_1 \geq \lambda_2 \geq 0$, we have
    \begin{align}
        & \max_{(R_0,R_1,R_2)\in \mathcal{M}(T)}\lambda_0 R_0+\lambda_1 R_1+\lambda_2 R_2   \nonumber 
        \\&\nonumber\quad= \min_{\alpha\in[0,1]}SR_{IB, 1}^{\lambda_0,\lambda_1,\lambda_2,\alpha}(T)\\& \quad=\min_{\alpha \in [0,1]}  \max_{q} \lambda_0 H_2(q) + q SR_{IB,1}^{\lambda_0,\lambda_1,\lambda_2,\alpha}(T_a) + \overline{q} SR_{IB,1}^{\lambda_0,\lambda_1,\lambda_2,\alpha}(T_b)
         \label{eq:keydec22}
         \\&\quad=\min_{\alpha \in [0,1]}\lambda_0\log\left(2^{\frac{1}{\lambda_0}SR_{IB,1}^{\lambda_0,\lambda_1,\lambda_2,\alpha}(T_a)} + 2^{\frac{1}{\lambda_0}SR_{IB,1}^{\lambda_0,\lambda_1,\lambda_2,\alpha}(T_b)}\right). \label{eq:keydec23}
    \end{align}
Similarly, for any $\lambda_0\geq\lambda_2\geq \lambda_1\geq 0$, we have
\begin{align}
        & \max_{(R_0,R_1,R_2)\in \mathcal{M}(T)}\lambda_0 R_0+\lambda_1 R_1+\lambda_2 R_2   \nonumber 
        \\&\nonumber \quad=\min_{\alpha\in[0,1]} SR_{IB, 2}^{\lambda_0,\lambda_1,\lambda_2,\alpha}(T)\\&\quad=\min_{\alpha \in [0,1]}  \max_{q} \lambda_0 H_2(q) + q SR_{IB,2}^{\lambda_0,\lambda_1,\lambda_2,\alpha}(T_a) + \overline{q} SR_{IB,2}^{\lambda_0,\lambda_1,\lambda_2,\alpha}(T_b)
         \label{eq:keydec22b}
         \\&\quad=\min_{\alpha \in [0,1]}\lambda_0\log\left(2^{\frac{1}{\lambda_0}SR_{IB,2}^{\lambda_0,\lambda_1,\lambda_2,\alpha}(T_a)} + 2^{\frac{1}{\lambda_0}SR_{IB,2}^{\lambda_0,\lambda_1,\lambda_2,\alpha}(T_b)}\right). \label{eq:keydec23b}
    \end{align}
\end{corollary}
\begin{proof}
    \eqref{eq:keydec22} and \eqref{eq:keydec22b} are direct results of \eqref{eqn:martonsumcharacterization1} and \eqref{eqn:martonsumcharacterization2};
    while \eqref{eq:keydec23} and \eqref{eq:keydec23b} are obtained by applying \eqref{rem:opt}.
\end{proof}

\section{The Gap between the UVW Outer Bound and Marton's Inner Bound}\label{sec:example_reversely-semi}

   \begin{figure}
\centering
\begin{minipage}{0.45\linewidth}
\resizebox{0.7\linewidth}{!}{%
\centering
\resizebox{40mm}{!}{%
\begin{tikzpicture}[
    >=stealth,
    xnode/.style={draw, circle, minimum size=7mm},
    ynode/.style={draw, circle, minimum size=7mm},
    znode/.style={draw, circle, minimum size=7mm},
    every edge/.style={->, thick}
]

\node[xnode] (x1) at (0,3) {$1$};
\node[xnode] (x2) at (0,1) {$2$};
\node[xnode] (x3) at (0,-1) {$3$};
\node[xnode] (x4) at (0,-3) {$4$};

\node[ynode] (y1) at (4,4.0) {$1$};
\node[ynode] (y2) at (4,3.0) {$2$};

\node[znode] (z1) at (4,1.0) {$1$};
\node[znode] (z2) at (4,0.0) {$2$};
\node[znode] (z3) at (4,-1.0) {$3$};
\node[znode] (z4) at (4,-2.0) {$4$};
\node[znode] (z5) at (4,-3.0) {$5$};
\node[znode] (z6) at (4,-4.0) {$6$};

\node[above=2mm] at (x1.north) {$X_a\in\{1,2,3,4\}$};
\node[above=2mm] at (y1.north) {$Z_a\in\{1,2\}$};
\node[below=2mm] at (z6.south) {$Y_a\in\{1,2,3,4,5,6\}$};

\draw[->, blue, thick] (x1) -- (y1);
\draw[->, blue, thick] (x2) -- (y1);
\draw[->, blue, thick] (x3) -- (y2);
\draw[->, blue, thick] (x4) -- (y2);

\draw[->, red, thick] (x1) --  (z1);
\draw[->, red, thick] (x1) --  (z2);
\draw[->, red, thick] (x1) -- (z3);

\draw[->, red, thick] (x2) -- (z1);
\draw[->, red, thick] (x2) -- (z4);
\draw[->, red, thick] (x2) -- (z5);

\draw[->, red, thick] (x3) -- (z2);
\draw[->, red, thick] (x3) -- (z4);
\draw[->, red, thick] (x3) -- (z6);

\draw[->, red, thick] (x4) -- (z3);
\draw[->, red, thick] (x4) -- (z5);
\draw[->, red, thick] (x4) -- (z6);

\end{tikzpicture}}}
\end{minipage}
\hfill
\begin{minipage}{0.45\linewidth}
\resizebox{0.7\linewidth}{!}{%
\centering
\resizebox{40mm}{!}{%
\begin{tikzpicture}[
    >=stealth,
    xnode/.style={draw, circle, minimum size=7mm},
    ynode/.style={draw, circle, minimum size=7mm},
    znode/.style={draw, circle, minimum size=7mm},
    every edge/.style={->, thick}
]

\node[xnode] (x1) at (0,3) {$1$};
\node[xnode] (x2) at (0,1) {$2$};
\node[xnode] (x3) at (0,-1) {$3$};
\node[xnode] (x4) at (0,-3) {$4$};

\node[ynode] (y1) at (4,4) {$1$};
\node[ynode] (y2) at (4,3) {$2$};

\node[znode] (z1) at (4,1.0) {$1$};
\node[znode] (z2) at (4,0.0) {$2$};
\node[znode] (z3) at (4,-1.0) {$3$};
\node[znode] (z4) at (4,-2.0) {$4$};
\node[znode] (z5) at (4,-3.0) {$5$};
\node[znode] (z6) at (4,-4.0) {$6$};

\node[above=2mm] at (x1.north) {$X_b\in\{1,2,3,4\}$};
\node[above=2mm] at (y1.north) {$Y_b\in\{1,2\}$};
\node[below=2mm] at (z6.south) {$Z_b\in\{1,2,3,4,5,6\}$};

\draw[->, red, thick] (x1) -- (y1);
\draw[->, red, thick] (x2) -- (y1);
\draw[->, red, thick] (x3) -- (y2);
\draw[->, red, thick] (x4) -- (y2);

\draw[->, blue, thick] (x1) --  (z1);
\draw[->, blue, thick] (x1) --  (z2);
\draw[->, blue, thick] (x1) -- (z3);

\draw[->, blue, thick] (x2) -- (z1);
\draw[->, blue, thick] (x2) -- (z4);
\draw[->, blue, thick] (x2) -- (z5);

\draw[->, blue, thick] (x3) -- (z2);
\draw[->, blue, thick] (x3) -- (z4);
\draw[->, blue, thick] (x3) -- (z6);

\draw[->, blue, thick] (x4) -- (z3);
\draw[->, blue, thick] (x4) -- (z5);
\draw[->, blue, thick] (x4) -- (z6);

\end{tikzpicture}}}
\end{minipage}

\caption{A reversely semi-deterministic sum broadcast channel.}
\label{fig:f1}
\end{figure}
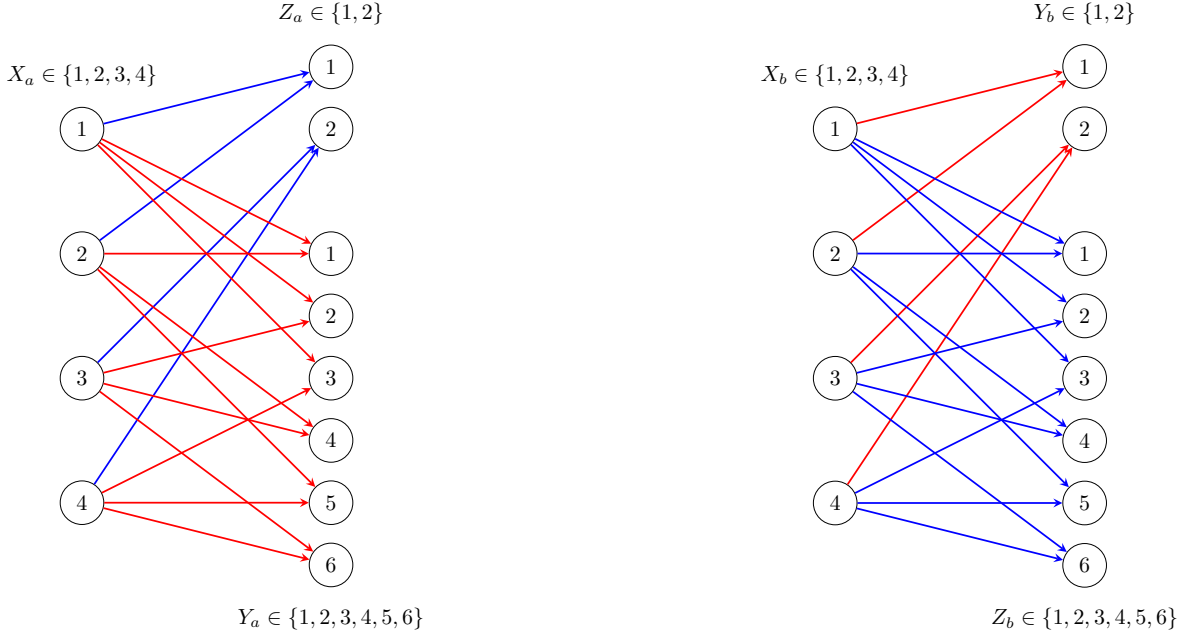

Consider the reversely semi-deterministic sum broadcast channel in Figure~\ref{fig:f1}. This channel resembles the example in \cite[Claim~3]{ggny14}. Assume that the transition probabilities are uniform over the possible outputs, i.e., the  red edges have probability $\frac{1}{3}$ in the first component and the blue edges have probability $\frac{1}{3}$ in the second component. By Lemma~\ref{lem:primaexam} and Lemma~\ref{lem:primaexam2}, $T_a, T_b\in \mathcal{P}^{\lambda_0,\lambda_1,\lambda_2}\cap \hat{\mathcal{P}}^{\lambda_0,\lambda_1,\lambda_2}$, since $T_a,T_b$ are both semi-deterministic. However, as we show in this section, the sum channel $T=T_a\oplus T_b$ does not fall into any primary class. Indeed, we show that the UVW sum-rate is greater than or equal to $5/2$, while Marton's sum-rate is $7/3$ for this channel.  Since $5/2 > 7/3$, this reversely semi-deterministic sum broadcast channel
exhibits a strict, non-zero gap between Marton’s inner bound and the UVW outer
bound. Consequently, this example illustrates that the UVW outer bound alone is
insufficient to establish capacity results for sum broadcast channels, even when
each component lies in $\mathcal{P}^{\lambda_0,\lambda_1,\lambda_2}$ or
$\hat{\mathcal{P}}^{\lambda_0,\lambda_1,\lambda_2}$, and thereby highlights the
novelty of the results established in this paper.

\subsection{The Sum-rate of Marton's inner bound}
In this section, we show that for the channel depicted in Figure~\ref{fig:f1}, we have
\begin{align*}
        \max_{(R_0,R_1,R_2)\in\mathcal{M}(T)} R_0+R_1+R_2 = \frac{7}{3}.
    \end{align*}

By setting $\lambda_0=\lambda_1=\lambda_2=1$, 
Marton's sum-rate can be evaluated, using Corollary~\ref{MartonCharacterizationSumChannel}, as
    \begin{align*}
        \max_{(R_0,R_1,R_2)\in\mathcal{M}(T)} R_0+R_1+R_2 \stackrel{(a)}=&        \min_{\alpha\in[0,1]}SR_{IB, 1}^{1,1,1,\alpha}(T)\\
        \stackrel{(b)}=&\min_{\alpha\in[0,1]}\log\left(2^{SR_{IB,1}^{1,1,1,\alpha}(T_a)} + 2^{SR_{IB,1}^{1,1,1,\alpha}(T_b)}\right).
    \end{align*}

    In \cite[Appendix~A and~B]{ggny14}, it is shown that, $SR_{IB,1}^{1,1,1,\alpha}(T_a)$ is achieved by a uniform input distribution, and evaluates to:
    \begin{align*}
        SR_{IB,1}^{1,1,1,\alpha}(T_a)= \begin{cases}
        \frac{5}{3}-\frac{2}{3}\alpha,& \alpha\in[0,\frac{1}{2}],\\
        \frac{4}{3},& \alpha\in(\frac{1}{2},1].
    \end{cases}
    \end{align*}
    Similarly,
    \begin{align*}
        SR_{IB,1}^{1,1,1,\alpha}(T_b)= \begin{cases}
        \frac{4}{3}, & \alpha\in[0,\frac{1}{2}],\\
        1+\frac{2}{3}\alpha& \alpha\in(\frac{1}{2},1].
    \end{cases}
    \end{align*}
    Thus,
    \begin{align*}
        \max_{(R_0,R_1,R_2)\in\mathcal{M}(T)} R_0+R_1+R_2 \stackrel{(a)}=&        \min_{\alpha\in[0,1]}SR_{IB, 1}^{1,1,1,\alpha}(T)\\
        \stackrel{(b)}=&\min_{\alpha\in[0,1]}\log\left(2^{SR_{IB,1}^{1,1,1,\alpha}(T_a)} + 2^{SR_{IB,1}^{1,1,1,\alpha}(T_b)}\right)\\
        =& \log\min\Bigg\{\min_{\alpha\in[0,1/2]}\left(2^{\frac{5}{3}-\frac{2}{3}\alpha} + 2^{ \frac{4}{3}}\right),
        \min_{\alpha\in[1/2,1]}\left(2^{\frac43} + 2^{ 1+\frac{2}{3}\alpha}\right)\Bigg\}
        \\ =& \frac{7}{3}.
    \end{align*}

\subsection{Lower bound on the sum-rate of the UVW outer bound}
In this section, we show that for the channel depicted in Figure~\ref{fig:f1},  the rate tuple $(R_0,R_1,R_2)=(0,5/4,5/4)$ belongs to the UVW outer bound. Consequently, 
we have
\begin{align*}
        \max_{(R_0,R_1,R_2)\in\mathcal{O}_{UVW}(T)} R_0+R_1+R_2 \geq \frac{5}{2}.
    \end{align*}

    Assume that $X_a$ and $X_b$ are independent and uniform on $\mathcal{X}_a$ and $\mathcal{X}_b$, respectively. Define $U_a,V_a,X_a,U_b,V_b,X_b$ with joint distribution of the form $p(x_a)p(x_b)p(u_a|x_a)p(v_a|x_a)p(u_b|x_b)p(v_b|x_b)$ as follows.
    
    Let $V_a = Z_a$ (so that $p(V_a=1)=p(V_a=2)=\frac{1}{2}$) and choose $p(u_a,x_a)$ such that
    \begin{align*}
        &p(X_a=1|U_a=1)=p(X_a=3|U_a=1)=\frac{1}{2},\\
        &p(X_a=2|U_a=1)=p(X_a=4|U_a=1)=0,\\&p(X_a=2|U_a=2)=p(X_a=4|U_a=2)=\frac{1}{2},\\
        &p(X_a=1|U_a=2)=p(X_a=3|U_a=2)=0,\\
        &p(U_a=1)=p(U_a=2)=\frac{1}{2}.
    \end{align*}
    Similarly, let $U_b= Y_b$ and choose $p(v_b,x_b)$ such that
    \begin{align*}
        &p(X_b=1|V_b=1)=p(X_b=3|V_b=1)=\frac{1}{2},\\
        &p(X_b=2|V_b=1)=p(X_b=4|V_b=1)=0,\\&p(X_b=2|V_b=2)=p(X_b=4|V_b=2)=\frac{1}{2},\\&p(X_b=1|V_b=2)=p(X_b=3|V_b=2)=0,\\
        &p(V_b=1)=p(V_b=2)=\frac{1}{2}.
    \end{align*}
    Let $Q\in\{a,b\}$ be a uniform Bernoulli random variable, independent of $U_a,V_a,X_a,U_b,V_b,X_b$. 
    Let
    \[
X = (Q, X_Q), \quad \tilde{U} = (Q, U_Q), \quad \tilde{V} = (Q, V_Q).
\]
Note that $I(Q;Y)=I(Q;Z)=H(Q)=1$, since $Q$ indicates the channel component and is therefore a function of either $Y$ or $Z$. 
    Let $\tilde W$ be a constant random variable. 
    
    We verify that the rate tuple $(R_0,R_1,R_2)=(0,5/4,5/4)$ belongs to the UVW outer bound for the choice of $(X, \tilde W, \tilde U,\tilde V)$. For $R_0$, we have
    \begin{align*}
        R_0\leq & \min\{I(\tilde W;Y),I(\tilde W;Z)\} = 0.
    \end{align*}
    
    For $R_0+R_1$, we have
    \begin{align*}
        R_0+R_1\leq &I(\tilde U,\tilde W;Y)\\
        =&I(U_Q,Q;Y)\\
        =&I(Q;Y)+\frac{1}{2} I(U_a;Y|Q=a)+\frac{1}{2} I(U_b;Y|Q=b)\\
        = & 1+\frac{1}{2} I(U_a;Y_a)+\frac{1}{2} I(U_b;Y_b)
        \\
        = & 1+\frac{1}{2}\cdot\frac13 +\frac{1}{2}\cdot 1\\
        =&\frac{5}{3},
    \end{align*}
     Similarly, one can show
    $$R_0+R_2\leq 1+\frac{1}{2} I(V_a;Z_a)+\frac{1}{2} I(V_b;Z_b)=\frac{5}{3}.$$
    For the sum-rate constraint, we obtain
    \begin{align*}
        R_0+R_1+R_2\leq & \min\{I(\tilde W;Y),I(\tilde W;Z)\}+I(\tilde U,\tilde W;Y)+I(X;Z|\tilde U,\tilde W)\\
        =&I(Q;Y)+\frac{1}{2}I(U_a;Y_a)+\frac{1}{2}I(U_b;Y_b)+\frac{1}{2}I(X_a;Z_a|U_a)+\frac{1}{2}I(X_b;Z_b|U_b)\\
        = &1+\frac{1}{2}\cdot \frac13+\frac{1}{2}\cdot 1 +\frac{1}{2}\cdot 1+\frac{1}{2}\cdot \frac23\\
        =&\frac{5}{2}.
    \end{align*}
    Similarly, one can compute
    \begin{align*}
        R_0+R_1+R_2\leq & \min\{I(\tilde W;Y),I(\tilde W;Z)\}+I(\tilde V,\tilde W;Z)+I(X;Y|\tilde V,\tilde W)=\frac{5}{2}.
    \end{align*}

    The rate tuple $(0,\frac{5}{4},\frac{5}{4})$ satisfies all the constraints and hence lies inside the rate region of UVW outer bound. As a result, sum-rate of the UVW outer bound for this channel is lower bounded by $0+\frac{5}{4}+\frac{5}{4}=\frac{5}{2}$.
    \section{Main Result}

\begin{theorem}\label{thm:maintheorem1} Let $T$ be a sum of two broadcast channels $T_a,T_b$, i.e., $T=T_a\oplus T_b$. Take some $\lambda_0\geq \max\{\lambda_1,\lambda_2\}\geq 0$. Assume that at least one of the following three conditions hold:
\vspace{0.2cm}
\begin{enumerate}    \renewcommand{\labelenumi}{(\alph{enumi})}
     \item 
      $T_a,T_b\in \mathcal{P}^{\lambda_0,\lambda_1,\lambda_2}$,\\
     \item $T_a\in\hat{\mathcal{P}}^{\lambda_0,\lambda_1,\lambda_2} $ and $T_b\in \mathcal{P}^{\lambda_0,\lambda_1,\lambda_2}$ (or $T_a\in\mathcal{P}^{\lambda_0,\lambda_1,\lambda_2}$ and $ T_b\in \hat{\mathcal{P}}^{\lambda_0,\lambda_1,\lambda_2}$),\\
     \item $T_a\in\hat{\mathcal{P}}^{\lambda_0,\lambda_1,\lambda_2}_{A}$ and $T_b\in \hat{\mathcal{P}}^{\lambda_0,\lambda_1,\lambda_2}_{B}$ (or $T_a\in\hat{\mathcal{P}}^{\lambda_0,\lambda_1,\lambda_2}_{B}$ and $T_b\in \hat{\mathcal{P}}^{\lambda_0,\lambda_1,\lambda_2}_{A}$).
\end{enumerate}
\vspace{0.2cm}
   Then,
   \begin{align*}
        \max_{(R_0,R_1,R_2)\in \mathcal{O}_{aux}(T)}\lambda_0 R_0+\lambda_1 R_1+\lambda_2 R_2&=\max_{(R_0,R_1,R_2)\in \mathcal{C}(T)}\lambda_0 R_0+\lambda_1 R_1+\lambda_2 R_2\\ &
        = \max_{(R_0,R_1,R_2)\in \mathcal{M}(T)}\lambda_0 R_0+\lambda_1 R_1+\lambda_2 R_2.
    \end{align*}
      In particular, if at least one of the above conditions holds for every choice of
weights $\lambda_0 \geq \max\{\lambda_1,\lambda_2\}\geq 0$, then the capacity
region is fully characterized and
\[
\mathcal{M}(T) = \mathcal{C}(T) = \mathcal{O}_{aux}(T).
\]
\end{theorem}

The proof of the above theorem is given in Section~\ref{thm4proofl}.

\begin{remark}
    We are not able to show that $\hat{\mathcal{P}}^{\lambda_0,\lambda_1,\lambda_2}_{A}$ and $\hat{\mathcal{P}}^{\lambda_0,\lambda_1,\lambda_2}_{B}$ are closed under the channel-sum operation, i.e., that a broadcast channel $T$ belongs to $\hat{\mathcal{P}}^{\lambda_0,\lambda_1,\lambda_2}_{A}$ whenever it can be written as the sum of broadcast channels $T_a,T_b\in\hat{\mathcal{P}}^{\lambda_0,\lambda_1,\lambda_2}_{A}$ (and similarly for $\hat{\mathcal{P}}^{\lambda_0,\lambda_1,\lambda_2}_{B}$). However, for the sum of two more-capable channels $T_a,T_b\in\hat{\mathcal{P}}^{\lambda_0,\lambda_1,\lambda_2}_{A}$, where $Z_a\overset{M.C.}{\succeq}Y_a$ and $Z_b\overset{M.C.}{\succeq}Y_b$, one can easily show that  $Z_a\oplus Z_b\overset{M.C.}{\succeq}Y_a\oplus Y_b$, so $T_a\oplus T_b$ is still a more-capable channel. Therefore, by Lemma~\ref{lem:primaexam2} and Lemma~\ref{lem:primcap}, $T_a\oplus T_b\in\hat{\mathcal{P}}^{\lambda_0,\lambda_1,\lambda_2}_{A}$, and Marton's inner bound is tight for this sum channel. \hfill $\diamondsuit$
\end{remark}
A direct consequence of Theorem~\ref{thm:maintheorem1}, Lemma~\ref{lem:primaexam}, and Lemma~\ref{lem:primaexam2} is the following: for a sum of broadcast channels, if both components belong to one of the following categories—degraded, less-noisy, more-capable, deterministic, or semi-deterministic broadcast channels—then Marton's inner bound $\mathcal{M}(T)$ and the auxiliary-receiver outer bound $\mathcal{O}_{aux}(T)$ coincide.
This result generalizes a previously known result for the case of the sum of two reversely degraded broadcast channels \cite{g80}.

\subsection{Some Useful Preliminaries}
\begin{lemma}
    Take a sum channel $T=T_a\oplus T_b$, and let $Q\in\{a,b\}$ indicate the component of the sum channel. Further denote $q:=\mathbb{P}(Q=a)$. Then, any achievable triple $(R_0,R_1,R_2)$ for the sum channel $T=T_a\oplus T_b$ lies in the intersection of the regions $\mathcal{O}_1$ and $\mathcal{O}_2$ defined as follows:
the region $\mathcal{O}_1$ is given by
\begin{subequations}\begin{align}
R_0 & \leq H_2(q)+\min \big\{ q I(W;Y_a)+ \overline q I(W^\dagger ; Y_b),~ q I(W; Z_a) + \overline q I(W^\dagger ;Z_b) \big \}, \label{eqn:o11}\\
R_0 + R_1 & \leq H_2(q) + q I(U,W; Y_a) + \overline q I(U^\dagger ,W^\dagger ; Y_b),\label{eqn:o12}\\
R_0 + R_2 & \leq H_2(q)+ q I(V, W;Z_a) + \overline q I(V^\dagger ,W^\dagger;Z_b),\label{eqn:o13}\\
R_0 + R_1 + R_2 & \leq H_2(q) + \min \{ q I(W; Y_a) + \overline q I(W^\dagger; Y_b),~ q I(W; Z_a) + \overline q I(W^\dagger; Z_b) \} \nonumber \\
&\qquad+ q \left(I(U;Y_a|W) + I(X_a;Z_a|U, W)\right)\nonumber \\
& \qquad + \overline{q} \min \big\{ I(U^\dagger;Y_b|W^\dagger) + I(X_b; Z_b|U^\dagger, W^\dagger),~   I(V^\dagger;Z_b|W^\dagger) +  I(X_b; Y_b|V^\dagger, W^\dagger) \big\},\label{eqn:o14} \\
R_0 + R_1 + R_2 & \leq H_2(q) + \min \{ q I(W; Y_a) + \overline q I(W^\dagger; Y_b),~ q I(W; Z_a) + \overline q I(W^\dagger; Z_b) \} \nonumber \\
&\qquad+ \overline q (I(V^\dagger;Z_b|W^\dagger) +I(X_b;Y_b|V^\dagger,W^\dagger))\nonumber \\
& \qquad + q \min \big\{ I(U;Y_a|W) + I(X_a; Z_a|U, W),~ I(V;Z_a|W)+I(X_a;Y_a|V,W) \big\},\label{eqn:o15}
\end{align}\end{subequations}
for some $q\in[0,1]$, $p(w, u,v,x_a)T_a(y_a,z_a|x_a)$ and $p(w^\dagger, u^\dagger, v^\dagger,x_b)T_b(y_b,z_b|x_b)$.

The region $\mathcal{O}_2$ is given by
 \begin{subequations}
\begin{align}
R_0 & \leq H_2(q)+\min \big\{ q I(W^\dagger;Y_a)+\overline q I(W; Y_b), q I(W^\dagger; Z_a) + \overline q I(W;Z_b) \big \}, \label{eqn:o21} \\
R_0 + R_1 & \leq H_2(q) + q I(U^\dagger,W^\dagger; Y_a) + \overline q I(U,W; Y_b), \label{eqn:o22} \\
R_0+R_2&\leq H_2(q)+qI(V^\dagger,W^\dagger;Z_a)+\overline q I(V,W;Z_b),\label{eqn:o23}\\
R_0 + R_1 + R_2 & \leq H_2(q) + \min \{q I(W^\dagger; Y_a) + \overline qI(W; Y_b),~ q I(W^\dagger; Z_a) + \overline q I(W; Z_b) \} \nonumber 
 \\
& \qquad +  \overline q( I(U;Y_b|W) + I(X_b; Z_b|U, W))\nonumber \\
&\qquad+q \min\{I(U^\dagger ;Y_a|W^\dagger)+I(X_a;Z_a|U^\dagger, W^\dagger),~ I(V^\dagger;Z_a|W^\dagger) + I(X_a;Y_a|V^\dagger, W^\dagger)\},
\label{eqn:o24}\\
R_0 + R_1 + R_2 & \leq H_2(q) + \min \{q I(W^\dagger; Y_a) + \overline qI(W; Y_b),~ q I(W^\dagger; Z_a) + \overline q I(W; Z_b) \} \nonumber \\
&\qquad+q(I(V^\dagger;Z_a|W^\dagger )+I(X_a;Y_a|V^\dagger, W^\dagger))\nonumber\\
&\qquad +\overline q \min\{I(U;Y_b|W)+I(X_b;Z_b|U,W),~I(V;Z_b|W)+I(X_b;Y_b|V,W)\}, \label{eqn:o25}
\end{align}
\end{subequations}
for some $q\in[0,1]$,  $p(w^\dagger, u^\dagger , v^\dagger ,x_a)T_a(y_a,z_a|x_a)$ and $p(w,u,v,x_b)T_b(y_b,z_b|x_b)$.\end{lemma}

\begin{proof}
    Consider the constraints \eqref{eqnSUM1a},\eqref{eqnSUM2a},\eqref{eqnSUM5a},\eqref{eqnSUM6a},\eqref{eqnSUM7a} in the outer bound region in Theorem~\ref{th:obp}. We derive the proposed outer bound by selecting $G = (Q,\tilde G)$ and $K=(Q,\tilde K)$ in Theorem~\ref{th:obp}, specifically:
\begin{itemize}
    \item For $\mathcal{O}_1$, we set $ \tilde G = \emptyset$ and $\tilde K = Z_a$ when $Q = a$, while $\tilde G = Y_b$ and $\tilde K = \emptyset$  when $Q = b$
    \item For $\mathcal{O}_2$, we set $ \tilde G = Y_a$ and $\tilde K = \emptyset$ when $Q = a$, while $\tilde G = \emptyset$ and $\tilde K = Z_b$  when $Q = b$
\end{itemize}
After substituting these choices, we then upper-bound the terms 
$I(W;Q)$ and $I(W^\dagger;Q)$ by $H(Q)$. We obtain \eqref{eqn:o11},\eqref{eqn:o12},\eqref{eqn:o13},\eqref{eqn:o14},\eqref{eqn:o15} for the first substitution choice, while we obtain \eqref{eqn:o21},\eqref{eqn:o22},\eqref{eqn:o23},\eqref{eqn:o24},\eqref{eqn:o25} for the second substitution choice. See Appendix~\ref{o1derivation} for a step-by-step derivation of $\mathcal{O}_1$ (derivation of $\mathcal{O}_2$ is similar).
\end{proof}

\begin{corollary} 
    Any achievable triple $(R_0,R_1,R_2)$ for the sum broadcast channel $T=T_a\oplus T_b$
    must satisfy, for any $\lambda_0 \geq \lambda_1 \geq \lambda_2\geq 0$,
 \begin{align}
        & \max_{(R_0,R_1,R_2)\in \mathcal{C}(T)}\lambda_0 R_0+\lambda_1 R_1+\lambda_2 R_2   \nonumber 
        \\& \qquad\qquad \quad\leq\min_{\alpha \in [0,1]}  \max_{q} \lambda_0 H_2(q) + q SR_{OB,1p}^{\lambda_0,\lambda_1,\lambda_2,\alpha}(T_a) + \overline{q} \min\{SR_{OB,1m}^{\lambda_0,\lambda_1,\lambda_2,\alpha}(T_b),SR_{OB,1n}^{\lambda_0,\lambda_1,\lambda_2,\alpha}(T_b) \},\label{eqn:obwsr1}\\
        & \max_{(R_0,R_1,R_2)\in \mathcal{C}(T)}\lambda_0 R_0+\lambda_1 R_1+\lambda_2 R_2   \nonumber 
        \\& \qquad \qquad \quad\leq\min_{\alpha \in [0,1]}  \max_{q} \lambda_0 H_2(q) + \overline q SR_{OB,1p}^{\lambda_0,\lambda_1,\lambda_2,\alpha}(T_b) + q \min\{SR_{OB,1m}^{\lambda_0,\lambda_1,\lambda_2,\alpha}(T_a),SR_{OB,1n}^{\lambda_0,\lambda_1,\lambda_2,\alpha}(T_a) \}.\label{eqn:obwsr2}
    \end{align}
And similarly, for any $\lambda_0 \geq \lambda_2 \geq \lambda_1\geq 0$,
    \begin{align}
         & \max_{(R_0,R_1,R_2)\in \mathcal{C}(T)}\lambda_0 R_0+\lambda_1 R_1+\lambda_2 R_2   \nonumber 
        \\& \qquad \qquad \quad\leq\min_{\alpha \in [0,1]}  \max_{q} \lambda_0 H_2(q) + q SR_{OB,2p}^{\lambda_0,\lambda_1,\lambda_2,\alpha}(T_a) + \overline{q} \min\{SR_{OB,2m}^{\lambda_0,\lambda_1,\lambda_2,\alpha}(T_b),SR_{OB,2n}^{\lambda_0,\lambda_1,\lambda_2,\alpha}(T_b) \},\label{eqn:obwsr3}\\
        & \max_{(R_0,R_1,R_2)\in \mathcal{C}(T)}\lambda_0 R_0+\lambda_1 R_1+\lambda_2 R_2   \nonumber 
        \\&\qquad \qquad  \quad\leq\min_{\alpha \in [0,1]}  \max_{q} \lambda_0 H_2(q) + \overline q SR_{OB,2p}^{\lambda_0,\lambda_1,\lambda_2,\alpha}(T_b) + q \min\{SR_{OB,2m}^{\lambda_0,\lambda_1,\lambda_2,\alpha}(T_a),SR_{OB,2n}^{\lambda_0,\lambda_1,\lambda_2,\alpha}(T_a) \}. \label{eqn:obwsr4}   
    \end{align}
\end{corollary}
\begin{proof}
    Fix some $\lambda_0 \geq \lambda_1 \geq \lambda_2\geq 0$. Recall the outer bound region $\mathcal{O}_1$. Multiplying \eqref{eqn:o11} by $(\lambda_0-\lambda_1)$, \eqref{eqn:o12} by $(\lambda_1-\lambda_2)$, and \eqref{eqn:o15} by $\lambda_2$ and summing yields:
\begin{align*}
    &\lambda_0 R_0+\lambda_1 R_1+\lambda_2 R_2\\
    &\quad \leq \lambda_0 H_2(q)+(\lambda_0-\lambda_1+\lambda_2)\min \big\{ q I(W;Y_a)+ \overline q I(W^\dagger ; Y_b),~ q I(W; Z_a) + \overline q I(W^\dagger ;Z_b) \big \}\\
    &\qquad +(\lambda_1-\lambda_2)\bigg(q I(U,W; Y_a) + \overline q I(U^\dagger ,W^\dagger ; Y_b)\bigg) +\lambda_2\bigg (\overline q (I(V^\dagger;Z_b|W^\dagger) +I(X_b;Y_b|V^\dagger,W^\dagger))\nonumber \\
& \qquad + q \min \big\{ I(U;Y_a|W) + I(X_a; Z_a|U, W),~ I(V;Z_a|W)+I(X_a;Y_a|V,W) \big\}\bigg )\\
&\quad \leq \lambda_0 H_2(q)+q\bigg((\lambda_0-\lambda_1+\lambda_2)(\alpha  I(W;Y_a)+\overline\alpha I(W; Z_a))+(\lambda_1-\lambda_2)I(U,W; Y_a)\\
&\qquad +\lambda_2 \min \big\{ I(U;Y_a|W) + I(X_a; Z_a|U, W),~ I(V;Z_a|W)+I(X_a;Y_a|V,W) \big\}\bigg)\\
&\qquad +\overline q\bigg((\lambda_0-\lambda_1+\lambda_2)(\alpha I(W^\dagger ; Y_b)+\overline\alpha I(W^\dagger ;Z_b))+(\lambda_1-\lambda_2)I(U^\dagger ,W^\dagger ; Y_b)\\
&\qquad+\lambda_2(I(V^\dagger;Z_b|W^\dagger) +I(X_b;Y_b|V^\dagger,W^\dagger))\bigg),
\end{align*}
for some $p_Q p(u,v,w,x_a)p(u^\dagger,v^\dagger,w^\dagger,x_b)$, and any $\alpha\in[0,1]$.

Therefore,
\begin{align}
    &\max_{(R_0,R_1,R_2)\in \mathcal{C}(T)} \lambda_0 R_0+\lambda_1 R_1+\lambda_2 R_2\nonumber    \\
    &\quad \leq \max_{(R_0,R_1,R_2)\in \mathcal{O}_{1}(T)} \lambda_0 R_0+\lambda_1 R_1+\lambda_2 R_2\nonumber    \\
    & \quad \leq \min_{\alpha \in [0,1]} \max_{q} \lambda_0 H_2(q) + q \bigg(\max_{p(u_a,v_a,w_a,x_a)} \alpha(\lambda_0-\lambda_1+\lambda_2) I(W_a;Y_a) + \overline{\alpha}(\lambda_0-\lambda_1+\lambda_2) I(W_a;Z_a)\nonumber\\
    &\qquad + (\lambda_1-\lambda_2) I(U_a,W_a;Y_a) +\lambda_2 \min\big\{I(U_a;Y_a|W_a)+I(X_a;Z_a|U_a,W_a), I(V_a;Z_a|W_a)+I(X_a;Y_a|V_a,W_a)\big\} \bigg)\nonumber\\
    &\qquad + \overline{q} \bigg( \max_{p(u_b,v_b,w_b,x_b)} \alpha (\lambda_0-\lambda_1+\lambda_2)I(W_b;Y_b) + \overline{\alpha}(\lambda_0-\lambda_1+\lambda_2) I(W_b;Z_b) + (\lambda_1 - \lambda_2) I(W_b, U_b;Y_b)\nonumber\\
    &\qquad + \lambda_2 (I(V_b;Z_b|W_b)+I(X_b;Y_b|V_b,W_b))\bigg)\nonumber\\
    &\quad \leq \min_{\alpha \in [0,1]} \max_{q} \lambda_0 H_2(q) + q SR_{OB,1p}^{\lambda_0,\lambda_1,\lambda_2,\alpha}(T_a) + \overline{q} SR_{OB,1n}^{\lambda_0,\lambda_1,\lambda_2,\alpha}(T_b)\nonumber\\
    &\quad = \min_{\alpha \in [0,1]}  \lambda_0 \log\left( 2^{\frac{1}{\lambda_0} SR_{OB,1p}^{\lambda_0,\lambda_1,\lambda_2,\alpha}(T_a)} + 2^{\frac{1}{\lambda_0} SR_{OB,1n}^{\lambda_0,\lambda_1,\lambda_2,\alpha}(T_b)}\right).\label{eqn:obwsrproof1}
\end{align}

Recall the outer bound region $\mathcal{O}_2$. Multiply \eqref{eqn:o21} by $(\lambda_0-\lambda_1)$, \eqref{eqn:o22} by $(\lambda_1-\lambda_2)$, and \eqref{eqn:o24} by $\lambda_2$ and summing yields
\begin{align*}
    &\lambda_0 R_0+\lambda_1R_1+\lambda_2 R_2\\
    &\quad\leq \lambda_0 H_2(q) + (\lambda_0-\lambda_1+\lambda_2)\min \big\{ q I(W^\dagger;Y_a)+\overline q I(W; Y_b), q I(W^\dagger; Z_a) + \overline q I(W;Z_b) \big \}\\
    &\qquad +(\lambda_1-\lambda_0)\bigg(q I(U^\dagger,W^\dagger; Y_a) + \overline q I(U,W; Y_b)\bigg)\\
    &\qquad +\lambda_2\bigg(\overline q( I(U;Y_b|W) + I(X_b; Z_b|U, W))+\\
& \quad\qquad\qquad +q \min\big \{I(U^\dagger ;Y_a|W^\dagger )+I(X_a;Z_a|U^\dagger, W^\dagger),~ I(V^\dagger;Z_a|W^\dagger) + I(X_a;Y_a|V^\dagger, W^\dagger)\big \}  \bigg)\\
&\quad\leq \lambda_0 H_2(q)+q \bigg(\alpha (\lambda_0-\lambda_1+\lambda_2)I(W^\dagger;Y_a) +\overline\alpha (\lambda_0-\lambda_1+\lambda_2)I(W^\dagger; Z_a)+(\lambda_1-\lambda_0)I(U^\dagger,W^\dagger;Y_a)\\
&\quad\qquad\qquad +\lambda_2 \min\big \{I(U^\dagger ;Y_a|W^\dagger)+I(X_a;Z_a|U^\dagger, W^\dagger),~ I(V^\dagger;Z_a|W^\dagger) + I(X_a;Y_a|V^\dagger, W^\dagger)\big \}\bigg)\\
&\qquad +\overline q \bigg(\alpha (\lambda_0-\lambda_1+\lambda_2)I(W; Y_b)+\overline\alpha (\lambda_0-\lambda_1+\lambda_2)I(W;Z_b)+(\lambda_1-\lambda_0)I(U,W;Y_b)\\
&\quad\qquad\qquad +\lambda_2(I(U;Y_b|W) + I(X_b; Z_b|U, W))\bigg),
\end{align*}
for some $p_Q p(u^\dagger,v^\dagger,w^\dagger,x_a) p(u,w,x_b)$, and any $\alpha \in [0,1]$.
Therefore,
\begin{align}
    &\max_{(R_0,R_1,R_2)\in \mathcal{C}(T)} \lambda_0 R_0+\lambda_1 R_1+\lambda_2 R_2\nonumber    \\
    &\quad \leq\max_{(R_0,R_1,R_2)\in \mathcal{O}_{2}(T)} \lambda_0 R_0+\lambda_1 R_1+\lambda_2 R_2\nonumber    \\
    & \quad \leq \min_{\alpha \in [0,1]} \max_{q} \lambda_0 H_2(q) + q \bigg(\max_{p(w_a,u_a,v_a,x_a)} \alpha (\lambda_0-\lambda_1+\lambda_2) I(W_a;Y_a) + \overline{\alpha} (\lambda_0-\lambda_1+\lambda_2)I(W_a;Z_a) \nonumber\\
    &\qquad\qquad+ (\lambda_1-\lambda_2) I(U_a,W_a;Y_a) +\lambda_2 \min\big \{I(U_a ;Y_a|W_a)+I(X_a;Z_a|U_a, W_a),~ I(V_a;Z_a|W_a) + I(X_a;Y_a|V_a, W_a)\big \} \bigg) \nonumber\\
    &\qquad + \overline{q} \bigg(\max_{p(w_b,u_b,x_b)}\alpha (\lambda_0-\lambda_1+\lambda_2)I(W_b;Y_b) + \overline{\alpha}(\lambda_0-\lambda_1+\lambda_2) I(W_b;Z_b) + (\lambda_1 - \lambda_2) I(U_b,W_b;Y_b)\nonumber\\
    &\qquad\qquad+\lambda_2( I(U_b;Y_b|W) + I(X_b; Z_b|U_b, W_b))\bigg) \nonumber\\
    &\quad \leq \min_{\alpha \in [0,1]} \max_{q} \lambda_0 H_2(q) + q SR_{OB,1p}^{\lambda_0,\lambda_1,\lambda_2,\alpha}(T_a) + \overline{q} SR_{OB,1m}^{\lambda_0,\lambda_1,\lambda_2,\alpha}(T_b)\nonumber\\
    &\quad = \min_{\alpha \in [0,1]}  \lambda_0 \log\left( 2^{\frac{1}{\lambda_0} SR_{OB,1p}^{\lambda_0,\lambda_1,\lambda_2,\alpha}(T_a)} + 2^{\frac{1}{\lambda_0} SR_{OB,1m}^{\lambda_0,\lambda_1,\lambda_2,\alpha}(T_b)}\right).\label{eqn:obwsrproof2}
\end{align}
Using \eqref{eqn:obwsrproof1} and \eqref{eqn:obwsrproof2}, we prove \eqref{eqn:obwsr1} as follows: 
\begin{align*}
    &\max_{(R_0,R_1,R_2)\in \mathcal{C}(T)} \lambda_0 R_0+\lambda_1 R_1+\lambda_2 R_2\\
    &\quad \leq \min\{\min_{\alpha \in [0,1]}  \lambda_0 \log\left( 2^{\frac{1}{\lambda_0} SR_{OB,1p}^{\lambda_0,\lambda_1,\lambda_2,\alpha}(T_a)} + 2^{\frac{1}{\lambda_0} SR_{OB,1n}^{\lambda_0,\lambda_1,\lambda_2,\alpha}(T_b)}\right),\min_{\alpha \in [0,1]}  \lambda_0 \log\left( 2^{\frac{1}{\lambda_0} SR_{OB,1p}^{\lambda_0,\lambda_1,\lambda_2,\alpha}(T_a)} + 2^{\frac{1}{\lambda_0} SR_{OB,1m}^{\lambda_0,\lambda_1,\lambda_2,\alpha}(T_b)}\right)\}\\
    &\quad =\min_{\alpha \in [0,1]}  \lambda_0 \log\left( 2^{\frac{1}{\lambda_0} SR_{OB,1p}^{\lambda_0,\lambda_1,\lambda_2,\alpha}(T_a)} + 2^{\frac{1}{\lambda_0} \min\{SR_{OB,1m}^{\lambda_0,\lambda_1,\lambda_2,\alpha}(T_b),~SR_{OB,1n}^{\lambda_0,\lambda_1,\lambda_2,\alpha}(T_b)\}}\right)\\
    &\quad=\min_{\alpha \in [0,1]} \max_{q} \lambda_0 H_2(q) + q SR_{OB,1p}^{\lambda_0,\lambda_1,\lambda_2,\alpha}(T_a) + \overline{q} \min\{SR_{OB,1m}^{\lambda_0,\lambda_1,\lambda_2,\alpha}(T_b),~SR_{OB,1n}^{\lambda_0,\lambda_1,\lambda_2,\alpha}(T_b)\}.
\end{align*}
Equations \eqref{eqn:obwsr2}, \eqref{eqn:obwsr3}, \eqref{eqn:obwsr4} are proved in a similar manner. 
\end{proof}

\subsection{Proof of Theorem \ref{thm:maintheorem1}}
\label{thm4proofl}
\subsubsection{Proof of part (a)}
Suppose $T = T_a\oplus T_b$, where $T_a,T_b\in \mathcal{P}^{\lambda_0,\lambda_1,\lambda_2}$.

Fix $\lambda_0 \geq \lambda_1 \geq \lambda_2 \geq 0 $. It suffices to prove
\begin{equation} \label{eq:mainthmdes} \max_{(R_0,R_1,R_2)\in \mathcal{M}(T)}\lambda_0 R_0+\lambda_1 R_1+\lambda_2 R_2 \geq \max_{(R_0,R_1,R_2)\in \mathcal{C}(T)}\lambda_0 R_0+\lambda_1 R_1+\lambda_2 R_2.\end{equation}
The other case where $\lambda_0 \geq \lambda_2 \geq \lambda_1\geq 0$ is handled analogously, by interchanging the roles of receivers and the corresponding auxiliary variables. 

\begin{align*}
    & \max_{(R_0,R_1,R_2)\in \mathcal{C}(T)}\lambda_0 R_0+\lambda_1 R_1+\lambda_2 R_2   \nonumber 
        \\& \quad\stackrel{(a)}\leq\min_{\alpha \in [0,1]}  \max_{q} \lambda_0 H_2(q) + q SR_{OB,1p}^{\lambda_0,\lambda_1,\lambda_2,\alpha}(T_a) + \overline{q} \min\{SR_{OB,1m}^{\lambda_0,\lambda_1,\lambda_2,\alpha}(T_b),SR_{OB,1n}^{\lambda_0,\lambda_1,\lambda_2,\alpha}(T_b) \}\\
        &\quad \stackrel{(b)}\leq \min_{\alpha \in [0,1]}  \max_{q} \lambda_0 H_2(q) + q \min\{SR_{OB,1m}^{\lambda_0,\lambda_1,\lambda_2,\alpha}(T_a),SR_{OB,1n}^{\lambda_0,\lambda_1,\lambda_2,\alpha}(T_a) \} + \overline{q} \min\{SR_{OB,1m}^{\lambda_0,\lambda_1,\lambda_2,\alpha}(T_b),SR_{OB,1n}^{\lambda_0,\lambda_1,\lambda_2,\alpha}(T_b) \}\\
        &\quad = \min_{\alpha \in [0,1]}  \lambda_0 \log\left( 2^{\frac{1}{\lambda_0} \min\{SR_{OB,1m}^{\lambda_0,\lambda_1,\lambda_2,\alpha}(T_a),~SR_{OB,1n}^{\lambda_0,\lambda_1,\lambda_2,\alpha}(T_a)\}} + 2^{\frac{1}{\lambda_0} \min\{SR_{OB,1m}^{\lambda_0,\lambda_1,\lambda_2,\alpha}(T_b),~SR_{OB,1n}^{\lambda_0,\lambda_1,\lambda_2,\alpha}(T_b)\}}\right)\\
        &\quad \stackrel{(c)}\leq \min_{\alpha \in [0,1]}  \lambda_0 \log\left( 2^{\frac{1}{\lambda_0} SR_{IB,1}^{\lambda_0,\lambda_1,\lambda_2,\alpha}(T_a)} + 2^{\frac{1}{\lambda_0} SR_{IB,1}^{\lambda_0,\lambda_1,\lambda_2,\alpha}(T_b)}\right)\\
        &\quad =\min_{\alpha \in [0,1]}  \max_{q} \lambda_0 H_2(q) + q SR_{IB,1}^{\lambda_0,\lambda_1,\lambda_2,\alpha}(T_a) + \overline{q} SR_{IB,1}^{\lambda_0,\lambda_1,\lambda_2,\alpha}(T_b)\\
        &\quad \stackrel{(d)}= \min_{\alpha\in [0,1]} SR_{IB,1}^{\lambda_0,\lambda_1,\lambda_2,\alpha}(T)\\
        &\quad \stackrel{(e)}=\max_{(R_0,R_1,R_2)\in \mathcal{M}(T)}\lambda_0 R_0+\lambda_1 R_1+\lambda_2 R_2,
\end{align*}
where $(a)$ follows by \eqref{eqn:obwsr1}; $(b)$ is by Lemma~\ref{lem:srobcompare}; $(c)$ holds by \eqref{eq:con1} and \eqref{eq:con2} as $T_a,T_b\in\mathcal{P}^{\lambda_0,\lambda_1,\lambda_2}$; $(d)$ results from \eqref{eqn:martonsumcharacterization1}. And $(e)$ results from Corollary~\ref{MartonCharacterizationSumChannel}.

\subsubsection{Proof of  part (b)}
Without loss of generality, assume $T_a\in \hat{\mathcal{P}}^{\lambda_0,\lambda_1,\lambda_2}$ and $T_b\in\mathcal{P}^{\lambda_0,\lambda_1,\lambda_2}$.

We first assume $T_a\in\hat{\mathcal{P}}^{\lambda_0,\lambda_1,\lambda_2}_{A}$. Then, for any $\lambda_0 \geq \lambda_1 \geq \lambda_2 \geq 0 $,
\begin{align*}
    & \max_{(R_0,R_1,R_2)\in \mathcal{C}(T)}\lambda_0 R_0+\lambda_1 R_1+\lambda_2 R_2   \nonumber 
        \\& \quad\stackrel{(a)}\leq\min_{\alpha \in [0,1]}  \max_{q} \lambda_0 H_2(q) + q SR_{OB,1p}^{\lambda_0,\lambda_1,\lambda_2,\alpha}(T_a) + \overline{q} \min\{SR_{OB,1m}^{\lambda_0,\lambda_1,\lambda_2,\alpha}(T_b),SR_{OB,1n}^{\lambda_0,\lambda_1,\lambda_2,\alpha}(T_b) \}\\
        &\quad = \min_{\alpha \in [0,1]}  \lambda_0 \log\left( 2^{\frac{1}{\lambda_0} SR_{OB,1p}^{\lambda_0,\lambda_1,\lambda_2,\alpha}(T_a)} + 2^{\frac{1}{\lambda_0} \min\{SR_{OB,1m}^{\lambda_0,\lambda_1,\lambda_2,\alpha}(T_b),SR_{OB,1n}^{\lambda_0,\lambda_1,\lambda_2,\alpha}(T_b) \}}\right)\\
        &\quad\stackrel{(b)}\leq \min_{\alpha \in [0,1]}  \lambda_0 \log\left( 2^{\frac{1}{\lambda_0} SR_{IB,1}^{\lambda_0,\lambda_1,\lambda_2,\alpha}(T_a)} + 2^{\frac{1}{\lambda_0} SR_{IB,1}^{\lambda_0,\lambda_1,\lambda_2,\alpha}(T_b)}\right)\\
        &\quad =\min_{\alpha \in [0,1]}  \max_{q} \lambda_0 H_2(q) + q SR_{IB,1}^{\lambda_0,\lambda_1,\lambda_2,\alpha}(T_a) + \overline{q} SR_{IB,1}^{\lambda_0,\lambda_1,\lambda_2,\alpha}(T_b)\\
        &\quad \stackrel{(c)}=\min_{\alpha \in [0,1]} SR_{IB,1}^{\lambda_0,\lambda_1,\lambda_2,\alpha}(T)\\
        &\quad \stackrel{(d)}=\max_{(R_0,R_1,R_2)\in \mathcal{M}(T)}\lambda_0 R_0+\lambda_1 R_1+\lambda_2 R_2,
\end{align*}
where $(a)$ follows by \eqref{eqn:obwsr1}; $(b)$ holds by \eqref{eq:con3} and \eqref{eq:con1} as $T_a\in\hat{\mathcal{P}}^{\lambda_0,\lambda_1,\lambda_2}_{A}$ and $T_b\in\mathcal{P}^{\lambda_0,\lambda_1,\lambda_2}$; $(c)$ results from \eqref{eqn:martonsumcharacterization1}. And $(d)$ results from Corollary~\ref{MartonCharacterizationSumChannel}.

Similarly, for any $\lambda_0 \geq \lambda_2 \geq \lambda_1 \geq 0 $, 
\begin{align*}
    & \max_{(R_0,R_1,R_2)\in \mathcal{C}(T)}\lambda_0 R_0+\lambda_1 R_1+\lambda_2 R_2   \nonumber 
        \\& \quad\stackrel{(a)}\leq\min_{\alpha \in [0,1]}  \max_{q} \lambda_0 H_2(q)  + q \min\{SR_{OB,2m}^{\lambda_0,\lambda_1,\lambda_2,\alpha}(T_a),SR_{OB,2n}^{\lambda_0,\lambda_1,\lambda_2,\alpha}(T_a) \}+ \overline q SR_{OB,2p}^{\lambda_0,\lambda_1,\lambda_2,\alpha}(T_b)\\
        &\quad\leq \min_{\alpha \in [0,1]}  \max_{q} \lambda_0 H_2(q)  + q SR_{OB,2m}^{\lambda_0,\lambda_1,\lambda_2,\alpha}(T_a)+ \overline q SR_{OB,2p}^{\lambda_0,\lambda_1,\lambda_2,\alpha}(T_b)\\
        & \quad\stackrel{(b)}\leq\min_{\alpha \in [0,1]}  \max_{q} \lambda_0 H_2(q) + q SR_{OB,2m}^{\lambda_0,\lambda_1,\lambda_2,\alpha}(T_a) + \overline{q} \min\{SR_{OB,2m}^{\lambda_0,\lambda_1,\lambda_2,\alpha}(T_b),~SR_{OB,2n}^{\lambda_0,\lambda_1,\lambda_2,\alpha}(T_b) \}\\
        &\quad = \min_{\alpha \in [0,1]}  \lambda_0 \log\left( 2^{\frac{1}{\lambda_0} SR_{OB,2m}^{\lambda_0,\lambda_1,\lambda_2,\alpha}(T_a)} + 2^{\frac{1}{\lambda_0} \min\{SR_{OB,2m}^{\lambda_0,\lambda_1,\lambda_2,\alpha}(T_b),~SR_{OB,2n}^{\lambda_0,\lambda_1,\lambda_2,\alpha}(T_b) \}}\right)\\
        &\quad\stackrel{(c)}\leq \min_{\alpha \in [0,1]}  \lambda_0 \log\left( 2^{\frac{1}{\lambda_0} SR_{IB,2}^{\lambda_0,\lambda_1,\lambda_2,\alpha}(T_a)} + 2^{\frac{1}{\lambda_0} SR_{IB,2}^{\lambda_0,\lambda_1,\lambda_2,\alpha}(T_b)}\right)\\
        &\quad =\min_{\alpha \in [0,1]}  \max_{q} \lambda_0 H_2(q) + q SR_{IB,2}^{\lambda_0,\lambda_1,\lambda_2,\alpha}(T_a) + \overline{q} SR_{IB,2}^{\lambda_0,\lambda_1,\lambda_2,\alpha}(T_b)\\
        &\quad \stackrel{(d)}=\min_{\alpha \in [0,1]} SR_{IB,2}^{\lambda_0,\lambda_1,\lambda_2,\alpha}(T)\\
        &\quad \stackrel{(e)}=\max_{(R_0,R_1,R_2)\in \mathcal{M}(T)}\lambda_0 R_0+\lambda_1 R_1+\lambda_2 R_2,
\end{align*}
where $(a)$ follows by \eqref{eqn:obwsr4}; $(b)$ follows by Lemma~\ref{lem:srobcompare}; $(c)$ holds by \eqref{eq:con4} and \eqref{eq:con2} as $T_a\in\hat{\mathcal{P}}^{\lambda_0,\lambda_1,\lambda_2}_{A}$ and $T_b\in\mathcal{P}^{\lambda_0,\lambda_1,\lambda_2}$; $(d)$ results from \eqref{eqn:martonsumcharacterization2}. And $(e)$ results from Corollary~\ref{MartonCharacterizationSumChannel}.

The case $T_a\in\hat{\mathcal{P}}^{\lambda_0,\lambda_1,\lambda_2}_{B},T_b\in\mathcal{P}^{\lambda_0,\lambda_1,\lambda_2}$ is handled analogously as above.

\subsubsection{Proof of  part (c)}
First suppose $T = T_a\oplus T_b$ where $T_a\in \hat{\mathcal{P}}^{\lambda_0,\lambda_1,\lambda_2}_{A}$ and $T_b\in\hat{\mathcal{P}}^{\lambda_0,\lambda_1,\lambda_2}_{B}$. Then for any $\lambda_0 \geq \lambda_1 \geq \lambda_2 \geq 0 $,
\begin{align*}
    & \max_{(R_0,R_1,R_2)\in \mathcal{C}(T)}\lambda_0 R_0+\lambda_1 R_1+\lambda_2 R_2   \nonumber 
        \\& \quad\stackrel{(a)}\leq\min_{\alpha \in [0,1]}  \max_{q} \lambda_0 H_2(q) + q SR_{OB,1p}^{\lambda_0,\lambda_1,\lambda_2,\alpha}(T_a) + \overline{q} \min\{SR_{OB,1m}^{\lambda_0,\lambda_1,\lambda_2,\alpha}(T_b),SR_{OB,1n}^{\lambda_0,\lambda_1,\lambda_2,\alpha}(T_b) \}\\
        & \quad\leq\min_{\alpha \in [0,1]}  \max_{q} \lambda_0 H_2(q) + q SR_{OB,1p}^{\lambda_0,\lambda_1,\lambda_2,\alpha}(T_a) + \overline{q} SR_{OB,1m}^{\lambda_0,\lambda_1,\lambda_2,\alpha}(T_b)\\
        &\quad = \min_{\alpha \in [0,1]}  \lambda_0 \log\left( 2^{\frac{1}{\lambda_0} SR_{OB,1p}^{\lambda_0,\lambda_1,\lambda_2,\alpha}(T_a)} + 2^{\frac{1}{\lambda_0} SR_{OB,1m}^{\lambda_0,\lambda_1,\lambda_2,\alpha}(T_b)}\right)\\
        &\quad\stackrel{(b)}\leq \min_{\alpha \in [0,1]}  \lambda_0 \log\left( 2^{\frac{1}{\lambda_0} SR_{IB,1}^{\lambda_0,\lambda_1,\lambda_2,\alpha}(T_a)} + 2^{\frac{1}{\lambda_0} SR_{IB,1}^{\lambda_0,\lambda_1,\lambda_2,\alpha}(T_b)}\right)\\
        &\quad =\min_{\alpha \in [0,1]}  \max_{q} \lambda_0 H_2(q) + q SR_{IB,1}^{\lambda_0,\lambda_1,\lambda_2,\alpha}(T_a) + \overline{q} SR_{IB,1}^{\lambda_0,\lambda_1,\lambda_2,\alpha}(T_b)\\
        &\quad \stackrel{(c)}=\min_{\alpha \in [0,1]}SR_{IB,1}^{\lambda_0,\lambda_1,\lambda_2,\alpha}(T)\\
        &\quad \stackrel{(d)}= \max_{(R_0,R_1,R_2)\in \mathcal{M}(T)}\lambda_0 R_0+\lambda_1 R_1+\lambda_2 R_2,
\end{align*}
where $(a)$ follows by \eqref{eqn:obwsr1}; $(b)$ holds by \eqref{eq:con3} and \eqref{eq:con5} as $T_a\in\hat{\mathcal{P}}^{\lambda_0,\lambda_1,\lambda_2}_{A}$ and $T_b\in\hat{\mathcal{P}}^{\lambda_0,\lambda_1,\lambda_2}_{B}$; $(c)$ results from \eqref{eqn:martonsumcharacterization1}. And $(d)$ results from Corollary~\ref{MartonCharacterizationSumChannel}.

Similarly, for any $\lambda_0 \geq \lambda_2 \geq \lambda_1 \geq 0 $,
\begin{align*}
    & \max_{(R_0,R_1,R_2)\in \mathcal{C}(T)}\lambda_0 R_0+\lambda_1 R_1+\lambda_2 R_2   \nonumber 
        \\& \quad\stackrel{(a)}\leq\min_{\alpha \in [0,1]}  \max_{q} \lambda_0 H_2(q)  + q \min\{SR_{OB,2m}^{\lambda_0,\lambda_1,\lambda_2,\alpha}(T_a),SR_{OB,2n}^{\lambda_0,\lambda_1,\lambda_2,\alpha}(T_a) \}+ \overline q SR_{OB,2p}^{\lambda_0,\lambda_1,\lambda_2,\alpha}(T_b)\\
        & \quad\leq\min_{\alpha \in [0,1]}  \max_{q} \lambda_0 H_2(q) + q SR_{OB,2m}^{\lambda_0,\lambda_1,\lambda_2,\alpha}(T_a) + \overline{q} SR_{OB,2p}^{\lambda_0,\lambda_1,\lambda_2,\alpha}(T_b)\\
        &\quad = \min_{\alpha \in [0,1]}  \lambda_0 \log\left( 2^{\frac{1}{\lambda_0} SR_{OB,2m}^{\lambda_0,\lambda_1,\lambda_2,\alpha}(T_a)} + 2^{\frac{1}{\lambda_0} SR_{OB,2p}^{\lambda_0,\lambda_1,\lambda_2,\alpha}(T_b)}\right)\\
        &\quad\stackrel{(b)}\leq \min_{\alpha \in [0,1]}  \lambda_0 \log\left( 2^{\frac{1}{\lambda_0} SR_{IB,2}^{\lambda_0,\lambda_1,\lambda_2,\alpha}(T_a)} + 2^{\frac{1}{\lambda_0} SR_{IB,2}^{\lambda_0,\lambda_1,\lambda_2,\alpha}(T_b)}\right)\\
        &\quad =\min_{\alpha \in [0,1]}  \max_{q} \lambda_0 H_2(q) + q SR_{IB,2}^{\lambda_0,\lambda_1,\lambda_2,\alpha}(T_a) + \overline{q} SR_{IB,2}^{\lambda_0,\lambda_1,\lambda_2,\alpha}(T_b)\\
        &\quad \stackrel{(c)}=\min_{\alpha \in [0,1]} SR_{IB,2}^{\lambda_0,\lambda_1,\lambda_2,\alpha}(T)\\
        &\quad \stackrel{(d)}=\max_{(R_0,R_1,R_2)\in \mathcal{M}(T)}\lambda_0 R_0+\lambda_1 R_1+\lambda_2 R_2,
\end{align*}
where $(a)$ follows by \eqref{eqn:obwsr4}; $(b)$ holds by \eqref{eq:con4} and \eqref{eq:con6} as $T_a\in\hat{\mathcal{P}}^{\lambda_0,\lambda_1,\lambda_2}_{A}$ and $T_b\in\hat{\mathcal{P}}^{\lambda_0,\lambda_1,\lambda_2}_{B}$; $(c)$ results from \eqref{eqn:martonsumcharacterization2}. And $(d)$ results from Corollary~\ref{MartonCharacterizationSumChannel}.

The case $T_a\in\hat{\mathcal{P}}^{\lambda_0,\lambda_1,\lambda_2}_{B},T_b\in\hat{\mathcal{P}}^{\lambda_0,\lambda_1,\lambda_2}_{A}$ is handled analogously as above.

Finally, since these  collections of hyperplanes characterize the convex capacity region $\mathcal{C}(T)$, this will complete the proof.

\section{Conclusion}
We investigated the sum channels of two-receiver broadcast channels. We defined two ``primary" classes of broadcast channels comprising previously studied classes, such as less-noisy, more-capable and semi-deterministic. Then, we established the capacity region for a class of sum-broadcast channels whose components are both ``primary". We utilized an outer bound using auxiliary receivers to fashion a converse. 

\bibliographystyle{amsplain}
\bibliography{mybiblio}
\appendices

\section{Proofs That Three Classes of Broadcast Channels Are Primary}\label{specialtogoodproof}
\subsection{More-Capable Broadcast Channels Belong to $\hat{\mathcal{P}}^{\lambda_0,\lambda_1,\lambda_2}$}\label{appendix:morecapable}

Suppose $Z\overset{M.C.}{\succeq} Y$. We first show that \eqref{eq:con3} holds. For any $\lambda_0 \geq \lambda_1 \geq \lambda_2 \geq 0 $, we have
    \begin{align*}
    &SR_{IB,1}^{\lambda_0,\lambda_1,\lambda_2,\alpha}(T) \\&\quad =\max_{p(w,u,v,x)}\alpha (\lambda_0-\lambda_1+\lambda_2)I(W;Y)+\overline\alpha (\lambda_0-\lambda_1+\lambda_2)I(W;Z)+(\lambda_1-\lambda_2)I(U,W;Y)\\
  &\qquad +\lambda_2(I(U;Y|W)+ I(V;Z|W)-I(U;V|W))\\
        &\quad \stackrel{(a)}\geq\max_{p(w,v,x)}\alpha (\lambda_0-\lambda_1+\lambda_2)I(W;Y)+\overline\alpha (\lambda_0-\lambda_1+\lambda_2)I(W;Z)+(\lambda_1-\lambda_2)I(W;Y)+ \lambda_2 I(V;Z|W)\\
    &\quad \stackrel{(b)}=\max_{p(w,x)}\alpha (\lambda_0-\lambda_1+\lambda_2)I(W;Y)+\overline\alpha (\lambda_0-\lambda_1+\lambda_2)I(W;Z)+(\lambda_1-\lambda_2)I(W;Y) + \lambda_2 I(X;Z|W)\\
    &\quad \stackrel{(c)}= \max_{p(w,u,x)}\alpha (\lambda_0-\lambda_1+\lambda_2)I(U,W;Y)+\overline\alpha (\lambda_0-\lambda_1+\lambda_2)I(U,W;Z)+(\lambda_1-\lambda_2)I(U,W;Y)\\
  &\qquad + \lambda_2 I(X;Z|U,W)\\
    &\quad \geq \max_{p(w,u,x)}\alpha (\lambda_0-\lambda_1+\lambda_2)I(W;Y)+\overline\alpha (\lambda_0-\lambda_1+\lambda_2)I(W;Z)\\
  &\qquad +(\lambda_1-\lambda_2)I(U,W;Y) + \lambda_2 (\alpha I(U;Y|W)+\overline\alpha I(U;Z|W)+I(X;Z|U,W))\\
  &\quad \stackrel{(d)}\geq \max_{p(w,u,x)}\alpha (\lambda_0-\lambda_1+\lambda_2)I(W;Y)+\overline\alpha (\lambda_0-\lambda_1+\lambda_2)I(W;Z)\\
  &\qquad +(\lambda_1-\lambda_2)I(U,W;Y) + \lambda_2\min\{I(U;Y|W)+ I(X;Z|U,W), ~I(X;Z|W)\}\\
  &\quad \stackrel{(e)}\geq \max_{p(w,u,v,x)}\alpha (\lambda_0-\lambda_1+\lambda_2)I(W;Y)+\overline\alpha (\lambda_0-\lambda_1+\lambda_2)I(W;Z)\\
  &\qquad +(\lambda_1-\lambda_2)I(U,W;Y) + \lambda_2\min\{I(U;Y|W)+ I(X;Z|U,W), ~I(V;Z|W)+I(X;Y|V,W)\}\\
    &\quad = SR_{OB,1p}^{\lambda_0,\lambda_1,\lambda_2,\alpha}(T).
\end{align*}
Step $(a)$ follows by setting $U$ to be a constant. Step $(b)$ holds since it is optimal to choose $V = X$. Step $(c)$ follows by replacing the auxiliary variable $W$ with the pair $(U,W)$. The inequality $(d)$ holds because \begin{align*}
    &\alpha I(U;Y|W)+\overline\alpha I(U;Z|W)+I(X;Z|U,W)\\
    =&\alpha(I(U;Y|W)+I(X;Z|U,W))+\overline \alpha(I(U;Z|W)+I(X;Z|U,W))\\
    \geq&\min\{I(U;Y|W)+I(X;Z|U,W),I(X;Z|W)\}.
\end{align*}
The more-capable property is applied in step $(e)$: for an arbitrary $V$ such that $V\mcchain X\mcchain (Y,Z)$ is Markov,
$$I(X;Z|W)= I(V;Z|W)+I(X;Z|V,W)\geq I(V;Z|W)+I(X;Y|V,W).$$
Next, we prove equation \eqref{eq:con4}. For any $\lambda_0 \geq \lambda_2 \geq \lambda_1 \geq 0 $, we have
\begin{align*}
    &SR_{IB,2}^{\lambda_0,\lambda_1,\lambda_2,\alpha}(T) \\
    &\quad=\max_{p(w,u,v,x)}\alpha (\lambda_0-\lambda_2+\lambda_1)I(W;Y)+\overline\alpha (\lambda_0-\lambda_2+\lambda_1)I(W;Z)+(\lambda_2-\lambda_1)I(V,W;Z)\\
  &\qquad+\lambda_1(I(U;Y|W)+ I(V;Z|W)-I(U;V|W))\\
  &\quad \stackrel{(a)}\geq \max_{p(w,v,x)}\alpha (\lambda_0-\lambda_2+\lambda_1)I(W;Y)+\overline\alpha (\lambda_0-\lambda_2+\lambda_1)I(W;Z)+(\lambda_2-\lambda_1)I(V,W;Z)+\lambda_1 I(V;Z|W)\\
  &\quad \stackrel{(b)}=\max_{p(w,x)}\alpha (\lambda_0-\lambda_2+\lambda_1)I(W;Y)+\overline\alpha (\lambda_0-\lambda_2+\lambda_1)I(W;Z)+(\lambda_2-\lambda_1)I(X;Z) +\lambda_1 I(X;Z|W)\\
  &\quad \stackrel{(c)}\geq \max_{p(w,v,x)}\alpha (\lambda_0-\lambda_2+\lambda_1)I(W;Y)+\overline\alpha (\lambda_0-\lambda_2+\lambda_1)I(W;Z)+(\lambda_2-\lambda_1)I(V,W;Z)\\
  &\qquad +\lambda_1 (I(V;Z|W)+I(X;Y|V,W))\\
  &\quad =SR_{OB,2m}^{\lambda_0,\lambda_1,\lambda_2,\alpha}(T).
\end{align*}
Step $(a)$ follows by setting $U$ to be a constant. Step $(b)$ holds since the maximum is achieved by $V=X$. For step $(c)$: for an arbitrary $V$ such that $V\mcchain X\mcchain (Y,Z)$ is Markov,
$I(X;Z)\geq I(V,W;Z)$. Moreover, the more-capable property is used here:
$$I(X;Z|W)= I(V;Z|W)+I(X;Z|V,W)\geq I(V;Z|W)+I(X;Y|V,W).$$
This completes the proof of $T\in\hat{\mathcal{P}}^{\lambda_0,\lambda_1,\lambda_2}_{A}$ when $Z\overset{M.C.}{\succeq} Y$.

We can prove $T\in\hat{\mathcal{P}}^{\lambda_0,\lambda_1,\lambda_2}_{B}$ when $Y\overset{M.C.}{\succeq} Z$ analogously by swapping $U$ with $V$ and $Y$ with $Z$ in the above proof argument.
\subsection{Less-Noisy Broadcast Channels Belong to $\mathcal{P}^{\lambda_0,\lambda_1,\lambda_2}\cap \hat{\mathcal{P}}^{\lambda_0,\lambda_1,\lambda_2}$}
We assume $Z\overset{L.N.}{\succeq} Y$. Since the less-noisy class is a subset of the more-capable class, we have $Z\overset{M.C.}{\succeq} Y$. Hence, by Appendix~\ref{appendix:morecapable}, for any $\lambda_0 \geq \lambda_2 \geq \lambda_1 \geq 0 $, we have 
$$SR_{IB,2}^{\lambda_0,\lambda_1,\lambda_2,\alpha}(T)\geq SR_{OB,2m}^{\lambda_0,\lambda_1,\lambda_2,\alpha}(T).$$

Moreover, for any $\lambda_0 \geq \lambda_1 \geq \lambda_2 \geq 0 $, we have
\begin{align*}
    &SR_{IB,1}^{\lambda_0,\lambda_1,\lambda_2,\alpha}(T)\\
    &\quad =\max_{p(w,u,v,x)} \alpha(\lambda_0-\lambda_1+\lambda_2) I(W;Y)+\overline\alpha (\lambda_0-\lambda_1+\lambda_2) I(W;Z)\\
    &\qquad +(\lambda_1-\lambda_2) I(U,W;Y)+\lambda_2\left(I(U;Y|W)+I(V;Z|W)-I(U;V|W)\right)\\
    &\quad \stackrel{(a)}{\geq}\max_{p(w,x)} \alpha(\lambda_0-\lambda_1+\lambda_2)I(W;Y)+\overline\alpha (\lambda_0-\lambda_1+\lambda_2)I(W;Z)\\
    &\qquad +(\lambda_1-\lambda_2)I(W;Y)+\lambda_2 I(X;Z|W)\\ 
    &\quad \stackrel{(b)}=\max_{p(w,u,x)} \alpha(\lambda_0-\lambda_1+\lambda_2)I(U,W;Y)+\overline\alpha (\lambda_0-\lambda_1+\lambda_2)I(U,W;Z)\\
    &\qquad +(\lambda_1-\lambda_2)I(U,W;Y)+\lambda_2 I(X;Z|U,W)\\
    &\quad\stackrel{(c)}\geq  \max_{p(w,u,v,x)} \alpha(\lambda_0-\lambda_1+\lambda_2) I(W;Y)+\overline\alpha (\lambda_0-\lambda_1+\lambda_2) I(W;Z)\\
    &\qquad +(\lambda_1-\lambda_2)I(U,W;Y) +\lambda_2\left(I(U;Y|W)+ I(X;Z|U,W)\right)\\
    &\quad =SR_{OB,1m}^{\lambda_0,\lambda_1,\lambda_2,\alpha}(T),
\end{align*}
where $(a)$ follows as we maximize over a smaller set of distributions where $V = X$ and $U$ is a constant random variable;  $(b)$ follows by substituting the auxiliary variable $W$ with the pair $(U,W)$; and $(c)$ uses the property of the less-noisy ordering such that $I(U;Z|W)\geq I(U;Y|W)$.

On the one hand, we conclude that less-noisy broadcast channels with $Z\overset{L.N.}{\succeq }Y$ belong to $\mathcal{P}^{\lambda_0,\lambda_1,\lambda_2}$, since
\begin{align*}
    &SR_{IB,1}^{\lambda_0,\lambda_1,\lambda_2,\alpha}(T)\geq SR_{OB,1m}^{\lambda_0,\lambda_1,\lambda_2,\alpha}(T)\geq \min\{SR_{OB,1m}^{\lambda_0,\lambda_1,\lambda_2,\alpha}(T),SR_{OB,1n}^{\lambda_0,\lambda_1,\lambda_2,\alpha}(T)\}, \quad \forall\lambda_0\geq\lambda_1\geq \lambda_2\geq 0,\\
    &SR_{IB,2}^{\lambda_0,\lambda_1,\lambda_2,\alpha}(T)\geq SR_{OB,2m}^{\lambda_0,\lambda_1,\lambda_2,\alpha}(T)\geq \min\{SR_{OB,2m}^{\lambda_0,\lambda_1,\lambda_2,\alpha}(T),SR_{OB,2n}^{\lambda_0,\lambda_1,\lambda_2,\alpha}(T)\}, \quad \forall\lambda_0\geq\lambda_2\geq \lambda_1\geq 0.
\end{align*}

On the other hand, we also conclude that less-noisy broadcast channels with $Z\overset{L.N.}{\succeq }Y$ belong to both $\hat{\mathcal{P}}^{\lambda_0,\lambda_1,\lambda_2}_{A}$ and $\hat{\mathcal{P}}^{\lambda_0,\lambda_1,\lambda_2}_{B}$, since
\begin{align*}
    &SR_{IB,1}^{\lambda_0,\lambda_1,\lambda_2,\alpha}(T)\geq SR_{OB,1m}^{\lambda_0,\lambda_1,\lambda_2,\alpha}(T) \stackrel{(a)}\geq SR_{OB,1p}^{\lambda_0,\lambda_1,\lambda_2,\alpha}(T), \quad \forall\lambda_0\geq\lambda_1\geq \lambda_2\geq 0,\\
    &SR_{IB,2}^{\lambda_0,\lambda_1,\lambda_2,\alpha}(T)\geq SR_{OB,2m}^{\lambda_0,\lambda_1,\lambda_2,\alpha}(T) \stackrel{(b)}\geq SR_{OB,2p}^{\lambda_0,\lambda_1,\lambda_2,\alpha}(T), \quad \forall\lambda_0\geq\lambda_2\geq \lambda_1\geq 0, 
\end{align*}
where $(a)$ and $(b)$ holds by Lemma~\ref{lem:srobcompare}.

The proof argument for less-noisy broadcast channels with $Y\overset{L.N.}{\succeq }Z$ follows analogously.

\subsection{Semi-Deterministic Broadcast Channels Belong to $\mathcal{P}^{\lambda_0,\lambda_1,\lambda_2}\cap\hat{\mathcal{P}}^{\lambda_0,\lambda_1,\lambda_2}$}
Assume $Z=f(X)$. For any $\lambda_0\geq\lambda_1\geq\lambda_2\geq 0$, we have
\begin{align*}
    &SR_{IB,1}^{\lambda_0,\lambda_1,\lambda_2,\alpha}(T)\\
    &\quad =\max_{p(w,u,v,x)} \alpha (\lambda_0-\lambda_1+\lambda_2)I(W;Y)+\overline\alpha (\lambda_0-\lambda_1+\lambda_2)I(W;Z)\\
    &\qquad +(\lambda_1-\lambda_2)I(U,W;Y)+\lambda_2(I(U;Y|W)+ I(V;Z|W)-I(U;V|W))\\
    &\quad \stackrel{(a)}{\geq}\max_{p(w,u,x)} \alpha (\lambda_0-\lambda_1+\lambda_2)I(W;Y)+\overline\alpha (\lambda_0-\lambda_1+\lambda_2)I(W;Z)\\
    &\qquad +(\lambda_1-\lambda_2)I(U,W;Y)+\lambda_2(I(U;Y|W)+ H(Z|U,W))\\
    &\quad =\max_{p(w,u,x)} \alpha (\lambda_0-\lambda_1+\lambda_2)I(W;Y)+\overline\alpha (\lambda_0-\lambda_1+\lambda_2)I(W;Z)\\
    &\qquad +(\lambda_1-\lambda_2)I(U,W;Y)+\lambda_2(I(U;Y|W)+ I(X;Z|U,W))\\
    &\quad=SR_{OB,1m}^{\lambda_0,\lambda_1,\lambda_2,\alpha}(T),
\end{align*}
where (a) holds by choosing $V=Z$.

Similarly, for any $\lambda_0\geq\lambda_2\geq\lambda_1\geq 0$, we have
\begin{align*}
    &SR_{IB,2}^{\lambda_0,\lambda_1,\lambda_2,\alpha}(T)\\
    &\quad =\max_{p(w,u,v,x)} \alpha (\lambda_0-\lambda_2+\lambda_1)I(W;Y)+\overline\alpha (\lambda_0-\lambda_2+\lambda_1)I(W;Z)\\
    &\qquad +(\lambda_2-\lambda_1)I(V,W;Z)+\lambda_1(I(U;Y|W)+ I(V;Z|W)-I(U;V|W))\\
    &\quad \stackrel{(a)}{\geq}\max_{p(w,u,v,x)} \alpha (\lambda_0-\lambda_2+\lambda_1)I(W;Y)+\overline\alpha (\lambda_0-\lambda_2+\lambda_1)I(W;Z)\\
    &\qquad +(\lambda_2-\lambda_1)H(Z)+\lambda_1(I(U;Y|W)+ H(Z|W)-I(U;Z|W))\\
    &\quad =\max_{p(w,u,x)} \alpha (\lambda_0-\lambda_2+\lambda_1)I(W;Y)+\overline\alpha (\lambda_0-\lambda_2+\lambda_1)I(W;Z)\\
    &\qquad +(\lambda_2-\lambda_1)H(Z)+\lambda_1(I(U;Y|W)+ H(Z|U,W))\\
    &\quad \geq \max_{p(w,u,x)} \alpha (\lambda_0-\lambda_2+\lambda_1)I(W;Y)+\overline\alpha (\lambda_0-\lambda_2+\lambda_1)I(W;Z)\\
    &\qquad +(\lambda_2-\lambda_1)I(V,W;Z)+\lambda_1(I(U;Y|W)+ I(X;Z|U,W))\\
    &\quad= SR_{OB,2n}^{\lambda_0,\lambda_1,\lambda_2,\alpha}(T).
\end{align*}
where (a) holds by choosing $V=Z$.

On one hand, we can conclude that semi-deterministic broadcast channels with $Z=f(X)$ for some function $f$ belong to $\mathcal{P}^{\lambda_0,\lambda_1,\lambda_2}$, since 
\begin{align*}
    &SR_{IB,1}^{\lambda_0,\lambda_1,\lambda_2,\alpha}(T)\geq SR_{OB,1m}^{\lambda_0,\lambda_1,\lambda_2,\alpha}(T)\geq \min\{SR_{OB,1m}^{\lambda_0,\lambda_1,\lambda_2,\alpha}(T),SR_{OB,1n}^{\lambda_0,\lambda_1,\lambda_2,\alpha}(T)\},\quad\forall\lambda_0\geq\lambda_1\geq \lambda_2\geq 0,\\
    &SR_{IB,2}^{\lambda_0,\lambda_1,\lambda_2,\alpha}(T)\geq SR_{OB,2n}^{\lambda_0,\lambda_1,\lambda_2,\alpha}(T)\geq \min\{SR_{OB,2m}^{\lambda_0,\lambda_1,\lambda_2,\alpha}(T),SR_{OB,2n}^{\lambda_0,\lambda_1,\lambda_2,\alpha}(T)\},\quad\forall\lambda_0\geq\lambda_2\geq \lambda_1\geq 0.
\end{align*}

On the other hand, we can also conclude that semi-deterministic broadcast channels with $Z=f(X)$ for some function $f$ belong to $\hat{\mathcal{P}}^{\lambda_0,\lambda_1,\lambda_2}_{B}$, since
\begin{align*}
    &SR_{IB,2}^{\lambda_0,\lambda_1,\lambda_2,\alpha}(T)\geq SR_{OB,2n}^{\lambda_0,\lambda_1,\lambda_2,\alpha}(T)\stackrel{(a)}\geq \min\{SR_{OB,2m}^{\lambda_0,\lambda_1,\lambda_2,\alpha}(T),SR_{OB,2n}^{\lambda_0,\lambda_1,\lambda_2,\alpha}(T)\},\quad\forall\lambda_0\geq\lambda_2\geq \lambda_1\geq 0,
\end{align*}
where $(a)$ follows by Lemma~\ref{lem:srobcompare}.

The proof argument for semi-deterministic broadcast channels with $Y=f(X)$ for some function $f$ follows analogously. They belong to $\mathcal{P}^{\lambda_0,\lambda_1,\lambda_2}$ and $\hat{\mathcal{P}}^{\lambda_0,\lambda_1,\lambda_2}_{A}$.

\subsection{Evaluation of \(SR_{IB,1}^{\lambda_0,\lambda_1,\lambda_2,\alpha}(T)\), \(SR_{OB,1m}^{\lambda_0,\lambda_1,\lambda_2,\alpha}(T)\), and \(SR_{OB,1n}^{\lambda_0,\lambda_1,\lambda_2,\alpha}(T)\) for Binary Input Broadcast Channels}
\label{sec:characti1}

In this section, we evaluate the sum-rate bounds \(SR_{IB,1}^{\lambda_0,\lambda_1,\lambda_2,\alpha}(T)\), \(SR_{OB,1m}^{\lambda_0,\lambda_1,\lambda_2,\alpha}(T)\), and \(SR_{OB,1n}^{\lambda_0,\lambda_1,\lambda_2,\alpha}(T)\) specifically for binary input broadcast channels. 

To facilitate our analysis, we utilize the upper concave envelope operator, denoted by \(\mathcal{C}_{p_X}[\cdot]\). This operator maps a function of the input distribution \(p_X\) to its upper concave envelope, providing a tractable formulation and bypasses the need for an explicit auxiliary variable representation, as discussed in \cite{nair2013upper}.

Take some $\lambda_0 \geq \lambda_1 \geq \lambda_2\geq 0$. For Marton's inner bound, the sum-rate can be expressed and simplified as follows: 
\begin{align*}
    SR_{IB,1}^{\lambda_0,\lambda_1,\lambda_2,\alpha}(T) &= \max_{p(w,u,v,x)} \Big[ \alpha (\lambda_0-\lambda_1+\lambda_2)I(W;Y)+\bar\alpha (\lambda_0-\lambda_1+\lambda_2)I(W;Z)\\
    &\quad +(\lambda_1-\lambda_2)I(U,W;Y)+\lambda_2\big(I(U;Y|W)+ I(V;Z|W)-I(U;V|W)\big) \Big]\\
    &= \max_{p_X} \Big[ (\lambda_0\alpha +(\lambda_1-\lambda_2)\bar \alpha)I(X;Y)+\bar\alpha (\lambda_0-\lambda_1+\lambda_2)I(X;Z)\\
    &\quad + \mathcal{C}_{p_X}\Big[-(\lambda_0\alpha +(\lambda_1-\lambda_2)\bar \alpha)I(X;Y)-\bar\alpha (\lambda_0-\lambda_1+\lambda_2)I(X;Z)\\
    &\qquad\quad +\max_{p_{U,V|X}}\big\{\lambda_1I(U;Y)+\lambda_2 I(V;Z)- \lambda_2 I(U;V)\big\}\Big] \Big]\\
    &\stackrel{(a)}{=} \max_{p_X} \Big[ (\lambda_0\alpha +(\lambda_1-\lambda_2)\bar \alpha)I(X;Y)+\bar\alpha (\lambda_0-\lambda_1+\lambda_2)I(X;Z)\\
    &\quad + \mathcal{C}_{p_X}\Big[-(\lambda_0\alpha +(\lambda_1-\lambda_2)\bar \alpha)I(X;Y)-\bar\alpha (\lambda_0-\lambda_1+\lambda_2)I(X;Z)\\
    &\qquad\quad +\max\big\{\lambda_1I(X;Y),\lambda_2 I(X;Z)\big\}\Big] \Big].
\end{align*}
Here, step \((a)\) follows directly from Theorem 3 in \cite{agn19} for $\lambda_0 \geq \lambda_1 \geq \lambda_2\geq 0$, which  simplifies the maximization over the auxiliary random variables \(U\) and \(V\). 

Similarly, by applying the concave envelope operator, one can directly verify the expression for the first outer bound:
\begin{align*}
    SR_{OB,1m}^{\lambda_0,\lambda_1,\lambda_2,\alpha}(T) &= \max_{p(w,u,x)} \Big[ \alpha (\lambda_0-\lambda_1+\lambda_2)I(W;Y)+\bar\alpha (\lambda_0-\lambda_1+\lambda_2)I(W;Z)\\
    &\quad +(\lambda_1-\lambda_2)I(U,W;Y) +\lambda_2(I(U;Y|W)+ I(X;Z|U,W)) \Big]\\
    &= \max_{p_X} \Big[ (\lambda_0\alpha +(\lambda_1-\lambda_2)\bar \alpha)I(X;Y)+\bar\alpha (\lambda_0-\lambda_1+\lambda_2)I(X;Z)\\
    &\quad + \mathcal{C}_{p_X}\Big[-(\lambda_0\alpha +(\lambda_1-\lambda_2)\bar \alpha)I(X;Y)-\bar\alpha (\lambda_0-\lambda_1+\lambda_2)I(X;Z)\\
    &\qquad\quad + \lambda_1 I(X;Y) + \mathcal{C}_{p_X}[ \lambda_2 I(X;Z) - \lambda_1 I(X;Y)]\Big] \Big].
\end{align*}

Following a parallel derivation, the second outer bound is given by:
\begin{align*}
    SR_{OB,1n}^{\lambda_0,\lambda_1,\lambda_2,\alpha}(T) &= \max_{p(w,v,x)} \Big[ \alpha (\lambda_0-\lambda_1+\lambda_2)I(W;Y)+\bar\alpha (\lambda_0-\lambda_1+\lambda_2)I(W;Z)\\
    &\quad +(\lambda_1-\lambda_2)I(X;Y) +\lambda_2(I(V;Z|W)+ I(X;Y|V,W)) \Big]\\ 
    &= \max_{p_X} \Big[ (\lambda_0\alpha +(\lambda_1-\lambda_2)\bar \alpha)I(X;Y)+\bar\alpha (\lambda_0-\lambda_1+\lambda_2)I(X;Z)\\
    &\quad + \mathcal{C}_{p_X}\Big[-\alpha (\lambda_0-\lambda_1+\lambda_2)I(X;Y)-\bar\alpha (\lambda_0-\lambda_1+\lambda_2)I(X;Z) + \lambda_2 I(X;Z)\\
    &\qquad\quad + \lambda_2 \mathcal{C}_{p_X}[  I(X;Y) -  I(X;Z)]\Big] \Big].
\end{align*}  

\subsection{A More-Capable Channel Outside the Primary Class \(\mathcal{P}^{\lambda_0,\lambda_1,\lambda_2}\)}
\label{sec:appndAE}

To show the existence of a more-capable channel outside the primary class \(\mathcal{P}^{\lambda_0,\lambda_1,\lambda_2}\), consider a broadcast channel where \(T_{Z|X} \sim \text{BEC}(\epsilon)\) and \(T_{Y|X} \sim \text{BSC}(p)\), coupled with the condition \(\epsilon=H_2(p)\). Under these parameters, it is a well-known result that \(Z\) is more-capable than \(Y\), denoted as \(Z\overset{M.C.}{\succeq }Y\) \cite{nair2009capacity}. Let \(\lambda_0=\lambda_1 = 1+\delta\) and \(\lambda_2=1\) for some \(\delta>0\). 

We use the characterizations given in Appendix~\ref{sec:characti1}.
    We begin by evaluating the inner bound:
    \begin{align*}
        SR_{IB,1}^{\lambda_0,\lambda_1,\lambda_2,\alpha}(T) &= \max_{p_X} \Big[ (\lambda_0\alpha +(\lambda_1-\lambda_2)\bar \alpha)I(X;Y)+\bar\alpha (\lambda_0-\lambda_1+\lambda_2)I(X;Z)\\
        &\quad + \mathcal{C}_{p_X}\Big[-(\lambda_0\alpha +(\lambda_1-\lambda_2)\bar \alpha)I(X;Y)-\bar\alpha (\lambda_0-\lambda_1+\lambda_2)I(X;Z)\\
        &\qquad\quad +\max\big\{\lambda_1I(X;Y),\lambda_2 I(X;Z)\big\}\Big] \Big].
    \end{align*}
    Due to the inherent symmetries of both the BEC and BSC channels, the linear combination \((\lambda_0\alpha +(\lambda_1-\lambda_2)\bar \alpha)I(X;Y)+\bar\alpha (\lambda_0-\lambda_1+\lambda_2)I(X;Z)\) is maximized when the input \(X\) follows a uniform distribution. 
    
    Furthermore, as established in \cite{nair2013upper}, the upper concave envelope of a symmetric function over these channels is also maximized at the uniform input distribution. Specifically, the envelope term:
    \begin{align*}
        &\mathcal{C}_{p_X}\Big[-(\lambda_0\alpha +(\lambda_1-\lambda_2)\bar \alpha)I(X;Y)-\bar\alpha (\lambda_0-\lambda_1+\lambda_2)I(X;Z) +\max\big\{\lambda_1I(X;Y),\lambda_2 I(X;Z)\big\}\Big]
    \end{align*}
    evaluated at the uniform distribution is simply equal to its maximum value over all \(p_X\):
    \begin{align*}
        &\max_{p_X}\Big[-(\lambda_0\alpha +(\lambda_1-\lambda_2)\bar \alpha)I(X;Y)-\bar\alpha (\lambda_0-\lambda_1+\lambda_2)I(X;Z) +\max\big\{\lambda_1I(X;Y),\lambda_2 I(X;Z)\big\}\Big].
    \end{align*}
    Substituting \(\lambda_0=\lambda_1 = 1+\delta\) and \(\lambda_2=1\), and letting \(C\) denote the capacity achieved at the uniform distribution, we find:
    \begin{align*}
        SR_{IB,1}^{\lambda_0,\lambda_1,\lambda_2,\alpha}(T) &= (1+\delta)C + \max_{p_X}\Big[-((1+\delta)\alpha +\delta \bar \alpha)I(X;Y)-\bar\alpha I(X;Z)+\max\big\{(1+\delta)I(X;Y), I(X;Z)\big\}\Big]\\
        &= (1+\delta)C + \max_{p_X} \max\{\bar \alpha(I(X;Y)-I(X;Z)), \alpha I(X;Z) - (\alpha + \delta)I(X;Y) \}\\
        &= (1+\delta)C + \max_{p_X} \{ \alpha I(X;Z) - (\alpha + \delta)I(X;Y) \}
    \end{align*}
where the last equality holds because the first term $\bar \alpha(I(X;Y)-I(X;Z))\leq 0$ for all $p_X$ as \(Z\overset{M.C.}{\succeq }Y\), while 
$$\max_{p_X} \max\{\bar \alpha(I(X;Y)-I(X;Z)), \alpha I(X;Z) - (\alpha + \delta)I(X;Y) \}\geq 0$$
as constant $X$ gives value $0$.

    Next, we evaluate the first outer bound, \(SR_{OB,1m}^{\lambda_0,\lambda_1,\lambda_2,\alpha}(T)\). By the same symmetry arguments, the maximizing input distribution is uniform. The expression simplifies as follows:
    \begin{align*}
        SR_{OB,1m}^{\lambda_0,\lambda_1,\lambda_2,\alpha}(T) &= (1+\delta)C + \max_{p_X}\Big[-(\alpha+\delta)I(X;Y)-\bar\alpha I(X;Z)+ (1+\delta) I(X;Y) 
        + \mathcal{C}_{p_X}[  I(X;Z) - (1+\delta) I(X;Y)]\Big] \\
        &= (1+\delta)C + \max_{p_X}\Big\{\bar \alpha(I(X;Y)-I(X;Z)) + \mathcal{C}_{p_X}[  I(X;Z) - (1+\delta) I(X;Y)]  \Big\}\\
        &= (1+\delta)C  + \max_{p_X} \Big\{I(X;Z) - (1+\delta) I(X;Y) \Big\}.
    \end{align*}
    The final equality holds because of two key observations at the uniform distribution:
    (i) \(I(X;Y)-I(X;Z) \leq 0\), with equality holding for uniform \(X\).
    (ii) The concave envelope satisfies the upper bound:
    \[
        \mathcal{C}_{p_X}[I(X;Z) - (1+\delta) I(X;Y)] \leq \max_{p_X} \{I(X;Z) - (1+\delta) I(X;Y)\}
    \]
    and achieves this maximum at uniform \(X\). Therefore, the entire expression inside the maximization is maximized at the uniform input distribution.

    Finally, we evaluate the second outer bound, \(SR_{OB,1n}^{\lambda_0,\lambda_1,\lambda_2,\alpha}(T)\). Applying similar symmetry and envelope arguments, we obtain:
    \begin{align*}
        SR_{OB,1n}^{\lambda_0,\lambda_1,\lambda_2,\alpha}(T) &= (1+\delta)C + \max_{p_X}\Big\{\alpha (I(X;Z)-I(X;Y)) +  \mathcal{C}_{p_X}[  I(X;Y) -  I(X;Z)]\Big\}\\
        &= (1+\delta)C + \alpha\max_{p_X} \Big\{I(X;Z)-I(X;Y)\Big\},
    \end{align*}  
    where the last step holds because $I(X;Y)-I(X;Z)\leq 0$ for all $p_X$, achieving equality when $X$ is a constant random variable. Therefore,
    $$\mathcal{C}_{p_X}[  I(X;Y) -  I(X;Z)]=0, \qquad\forall p(x).$$
    To concretely demonstrate the separation between these bounds, consider a numerical instantiation where we fix \(p = 0.2\), \(\alpha = 0.5\), and \(\delta=0.1\). This yields:
    \begin{align*}
        C         &= 0.27807191,\\
        \epsilon  &= H_2(p) = 0.72192809,\\
        SR_{IB,1}^{\lambda_0,\lambda_1,\lambda_2,\alpha}   &= 0.31228936,\\
        SR_{OB,1m}^{\lambda_0,\lambda_1,\lambda_2,\alpha}  &= 0.32392242,\\
        SR_{OB,1n}^{\lambda_0,\lambda_1,\lambda_2,\alpha}  &= 0.31875207.
    \end{align*}
    Observe that both outer bounds, \(SR_{OB,1m}\) and \(SR_{OB,1n}\), are strictly greater than the inner bound \(SR_{IB,1}\). Because the bounds do not coincide, we conclude that this more-capable channel does not belong to the primary class \(\mathcal{P}^{\lambda_0,\lambda_1,\lambda_2}\).

\subsection{MIMO Gaussian Channels}
\label{sec:appndAF}

In the MIMO Gaussian setting, the channel outputs are given by $\bm{Y}=\bm{A}\bm{X}+\bm{G}_1$ and $\bm{Z}=\bm{B}\bm{X} + \bm{G}_2$, where the noise vectors $\bm{G}_1,\bm{G}_2 \sim \mathcal{N}(\bm{0},\bm{I})$. 

Geng and Nair \cite{gn14} showed that the maximum of the outer bound expression, subject to the covariance constraint $\mathbb{E}[\bm{X}\bm{X}^T] \preceq \bm{K}$,
\begin{align*}
    SR_{OB,1m}^{\lambda_0,\lambda_1,\lambda_2,\alpha}(T) =& \max_{p(w,u,\bm{x}): \mathbb{E}[\bm{X}\bm{X}^T] \preceq \bm{K}} \Big[ \alpha (\lambda_0-\lambda_1+\lambda_2)I(W;\bm{Y})+\bar\alpha (\lambda_0-\lambda_1+\lambda_2)I(W;\bm{Z})\\
    &\quad +(\lambda_1-\lambda_2)I(U,W;\bm{Y})+\lambda_2\big(I(U;\bm{Y}|W)+ I(\bm{X};\bm{Z}|U,W)\big) \Big]
\end{align*}
is attained by jointly Gaussian random variables. This is established using the doubling-followed-by-rotation argument. 

Furthermore, by Costa's Dirty Paper Coding (DPC) \cite{costa1983writing} principle, there exists a Gaussian auxiliary random variable $V$ correlated with $(U,\bm{X})$ such that:
\begin{align*}
    I(\bm{X};\bm{Z}|U,W) = I(V;\bm{Z}|W) - I(V;U|W).
\end{align*}

Substituting this identity into the expression for $SR_{OB,1m}^{\lambda_0,\lambda_1,\lambda_2,\alpha}(T)$ and relaxing the optimization domain to include $V$, we obtain:
\begin{align*}
    SR_{OB,1m}^{\lambda_0,\lambda_1,\lambda_2,\alpha}(T) =& \max_{p(w,u,\bm{x}): \mathbb{E}[\bm{X}\bm{X}^T] \preceq \bm{K}} \Big[ \alpha (\lambda_0-\lambda_1+\lambda_2)I(W;\bm{Y})+\bar\alpha (\lambda_0-\lambda_1+\lambda_2)I(W;\bm{Z})\\
    &\quad +(\lambda_1-\lambda_2)I(U,W;\bm{Y})+\lambda_2\big(I(U;\bm{Y}|W)+ I(\bm{X};\bm{Z}|U,W)\big) \Big]\\
    \leq & \max_{p(w,u,v,\bm{x}): \mathbb{E}[\bm{X}\bm{X}^T] \preceq \bm{K}} \Big[ \alpha (\lambda_0-\lambda_1+\lambda_2)I(W;\bm{Y})+\bar\alpha (\lambda_0-\lambda_1+\lambda_2)I(W;\bm{Z})\\
    &\quad +(\lambda_1-\lambda_2)I(U,W;\bm{Y})+\lambda_2\big(I(U;\bm{Y}|W)+ I(V;\bm{Z}|W)-I(V;U|W)\big) \Big]\\
    =& SR_{IB,1}^{\lambda_0,\lambda_1,\lambda_2,\alpha}(T).
\end{align*}
 This shows that $SR_{IB,1}^{\lambda_0,\lambda_1,\lambda_2,\alpha}(T)\geq SR_{OB,1m}^{\lambda_0,\lambda_1,\lambda_2,\alpha}(T)$ for all $\alpha\in[0,1]$, i.e., condition \eqref{eq:con1} holds. Condition \eqref{eq:con2} follows analogously by swapping the roles of $U$ and $V$ in the above argument. Together, these imply that MIMO Gaussian broadcast channels belong to the primary class $\mathcal{P}^{\lambda_0,\lambda_1,\lambda_2}$.

\subsection{Detailed Derivation for Outer Bound $\mathcal{O}_1$}\label{o1derivation}
To derive $\mathcal{O}_1$, we choose $T_{G,K|X,Y,Z}$ as $G = (Q,\tilde G),K=(Q,\tilde K)$, where $\tilde G = \emptyset$ and $\tilde K = Z_a$ when $Q = a$; while $\tilde G = Y_b$ and $\tilde K = \emptyset$  when $Q = b$.

The constraint \eqref{eqnSUM1a} becomes:
\begin{align*}
    R_0\leq &\min\{I(W^\dagger;G)+I(W;Y|G),~I(W^\dagger;Z|K)+I(W;K)\}\\
    =&\min\{I(W^\dagger;Q,\tilde G)+I(W;Y|Q,\tilde G),~I(W^\dagger;Z|Q,\tilde K)+I(W;Q,\tilde K)\}\\
    =&\min\{I(W^\dagger;Q)+I(W^\dagger;\tilde G|Q)+I(W;Y|Q,\tilde G),~I(W;Q)+I(W^\dagger;Z|Q,\tilde K)+I(W;\tilde K|Q)\}\\
    =&\min\{I(W^\dagger;Q)+\overline q I(W^\dagger;Y_b|Q=b)+q I(W;Y_a|Q=a),~I(W;Q)+\overline q I(W^\dagger;Z_b|Q=b)+q I(W;Z_a|Q=a)\}\\
    \leq&\min\{\overline q I(W^\dagger;Y_b|Q=b)+q I(W;Y_a|Q=a),~\overline q I(W^\dagger;Z_b|Q=b)+q I(W;Z_a|Q=a)\}+H_2(q).
\end{align*}
The constraint \eqref{eqnSUM2a} becomes:
\begin{align*}
    R_0 + R_1 & \leq I(U^\dagger,W^\dagger; G) + I(U,W; Y|G)\\
    &=I(U^\dagger,W^\dagger; Q,\tilde G) + I(U,W; Y|Q,\tilde G)\\
    &=I(U^\dagger,W^\dagger; Q)+I(U^\dagger,W^\dagger; \tilde G|Q) + I(U,W; Y|Q,\tilde G)\\
    &=I(U^\dagger,W^\dagger; Q)+\overline q I(U^\dagger,W^\dagger; Y_b|Q=b) + q I(U,W; Y_a|Q=a)\\
    &\leq H( Q)+\overline q I(U^\dagger,W^\dagger; Y_b|Q=b) + q I(U,W; Y_a|Q=a).
\end{align*}
The constraint \eqref{eqnSUM5a} becomes:
\begin{align*}
    R_0 + R_2 & \leq 
I(V,W; K) + I(V^\dagger,W^\dagger; Z|K)\\
    &=I(V,W; Q,\tilde K) + I(V^\dagger,W^\dagger; Z|Q, \tilde K)\\
    &=I(V,W; Q)+I(V,W;\tilde K|Q) + I(V^\dagger,W^\dagger; Z|Q, \tilde K)\\
    &=I(V,W; Q)+q I(V,W;Z_a|Q=a) + \overline q I(V^\dagger,W^\dagger; Z_b|Q=b)\\
    &\leq H(Q)+q I(V,W;Z_a|Q=a) + \overline q I(V^\dagger,W^\dagger; Z_b|Q=b).
\end{align*}
The constraint \eqref{eqnSUM6a} becomes:
\begin{align*}
    R_0 + R_1 + R_2 & \leq \min \{ I(W^\dagger, K; G) + I(W; Y|G),~ I(W^\dagger; Z|K) + I(W, G; K) \} \\
&\quad+ I(U;Y|W, G) + I(X;K|U, W, G) \\
& \quad + \min \big\{ I(U^\dagger;G|W^\dagger, K) + I(X; Z|U^\dagger, W^\dagger, K),~ I(V^\dagger;Z|W^\dagger, K) + I(X; G|V^\dagger, W^\dagger, K) \big\}\\
    &=\min \{ I(W^\dagger, Q,\tilde K; Q,\tilde G) + I(W; Y|Q,\tilde G),~ I(W^\dagger; Z|Q,\tilde K) + I(W, Q,\tilde G; Q,\tilde K) \} \\
&\quad+ I(U;Y|W, Q,\tilde G) + I(X;Q,\tilde K|U, W, Q,\tilde G) \\
& \quad + \min \big\{ I(U^\dagger;Q,\tilde G|W^\dagger, Q,\tilde K) + I(X; Z|U^\dagger, W^\dagger, Q,\tilde K),~ I(V^\dagger;Z|W^\dagger, Q,\tilde K) + I(X; Q,\tilde G|V^\dagger, W^\dagger, Q,\tilde K) \big\}\\
&=H(Q)+\min \{ I(W^\dagger,\tilde K; \tilde G|Q) + I(W; Y|Q,\tilde G),~ I(W^\dagger; Z|Q,\tilde K) + I(W, \tilde G; \tilde K|Q) \} \\
&\quad+ I(U;Y|W, Q,\tilde G) + I(X;\tilde K|U, W, Q,\tilde G) \\
& \quad + \min \big\{ I(U^\dagger;\tilde G|W^\dagger, Q,\tilde K) + I(X; Z|U^\dagger, W^\dagger, Q,\tilde K),~ I(V^\dagger;Z|W^\dagger, Q,\tilde K) + I(X; \tilde G|V^\dagger, W^\dagger, Q,\tilde K) \big\}\\
&=H_2(q)+\min \{ \overline q I(W^\dagger; Y_b|Q=b) + q I(W; Y_a|Q=a),~ \overline q I(W^\dagger; Z_b|Q=b) + q I(W; Z_a|Q=a) \} \\
&\quad+ q (I(U;Y_a|W, Q=a) + I(X_a;Z_a|U, W, Q=a)) \\
& \quad +\overline q \min \big\{ I(U^\dagger;Y_b|W^\dagger, Q=b) + I(X_b; Z_b|U^\dagger, W^\dagger, Q=b),~ I(V^\dagger;Z_b|W^\dagger, Q=b) + I(X_b; Y_b|V^\dagger, W^\dagger, Q=b) \big\}.
\end{align*}
The constraint \eqref{eqnSUM7a} becomes:
\begin{align*}
    R_0 + R_1 + R_2 & \leq \min \{ I(W^\dagger, K; G) + I(W; Y|G),~ I(W^\dagger; Z|K) + I(W, G; K) \} \\
&\quad + I(V^\dagger;Z|W^\dagger, K) + I(X;G|V^\dagger, W^\dagger, K) \\&\quad+ \min \big\{ I(U;Y|W, G) + I(X; K|U, W, G),~I(V;K|W, G) + I(X; Y|V, W, G) \big\}\\
    &=\min \{ I(W^\dagger, Q,\tilde K; Q,\tilde G) + I(W; Y|Q,\tilde G),~ I(W^\dagger; Z|Q,\tilde K) + I(W, Q,\tilde G;Q,\tilde K) \} \\
&\quad + I(V^\dagger;Z|W^\dagger, Q,\tilde K) + I(X;Q,\tilde G|V^\dagger, W^\dagger, Q,\tilde K) \\&\quad+ \min \big\{ I(U;Y|W, Q,\tilde G) + I(X; Q,\tilde K|U, W, Q,\tilde G),~I(V;Q,\tilde K|W,Q,\tilde G) + I(X; Y|V, W,Q,\tilde G) \big\}\\
&=H(Q)+\min \{ I(W^\dagger,\tilde K;\tilde G|Q) + I(W; Y|Q,\tilde G),~ I(W^\dagger; Z|Q,\tilde K) + I(W,\tilde G;\tilde K|Q) \} \\
&\quad + I(V^\dagger;Z|W^\dagger, Q,\tilde K) + I(X;\tilde G|V^\dagger, W^\dagger, Q,\tilde K) \\&\quad+ \min \big\{ I(U;Y|W, Q,\tilde G) + I(X;\tilde K|U, W, Q,\tilde G),~I(V;\tilde K|W, Q,\tilde G) + I(X; Y|V, W,Q,\tilde G) \big\}\\
&=H_2(q)+\min \{ \overline q I(W^\dagger;Y_b|Q=b) + q I(W; Y_a|Q=a),~ \overline q I(W^\dagger; Z_b|Q=b) +q I(W;Z_a|Q=a) \} \\
&\quad +\overline q( I(V^\dagger;Z_b|W^\dagger, Q=b) + I(X_b;Y_b|V^\dagger, W^\dagger, Q=b)) \\
&\quad+ q\min \big\{ I(U;Y_a|W, Q=a) + I(X_a;Z_a|U, W, Q=a),~I(V;Z_a|W, Q=a) + I(X_a; Y_a|V, W,Q=a) \big\}.
\end{align*}
The rate region defined by the above five constraints is equivalent to $\mathcal{O}_1$ by considering the distributions 
$$(\tilde{X}_a,\tilde{Y}_a, \tilde{Z}_a, U_a,V_a,W_a)\sim p(x_a,y_a,z_a,u,v,w|Q=a)$$
and
$$(\tilde{X}_b,\tilde{Y}_b, \tilde{Z}_b, U^\dagger_b,V^\dagger_b,W^\dagger_b)\sim p(x_b,y_b,z_b,u^\dagger,v^\dagger,w^\dagger|Q=b)$$
in the definition of $\mathcal{O}_1$.
\end{document}